\documentclass[article,onecolumn]{IEEEtran}


\usepackage[utf8]{inputenc} 
\usepackage[T1]{fontenc}
\usepackage{graphicx,colordvi,psfrag}
\usepackage{amsmath,amssymb}
\usepackage{epstopdf}
\usepackage[caption=false]{subfig}
\usepackage{epsfig,cite}
\usepackage{calc,pstricks, pgf, xcolor}
\usepackage{nicefrac}
\usepackage{enumerate}
\usepackage{dsfont}
\usepackage{bm}
\usepackage{color}
\usepackage{bbold}
\usepackage{setspace}

\newtheorem{theorem}{Theorem}
\newtheorem{corollary}{Corollary}
\newtheorem{remark}{Remark}

\newtheorem{proposition}{Proposition}

\newtheorem{conjecture}{Conjecture}
\newenvironment{proof}[1][Proof]{\noindent\textbf{#1.} }{\ \rule{0.5em}{0.5em}}

\newcommand{\Ind}{\mathds{1}}

\newcommand{\ba}{\mathbf{a}}

\newcommand{\RR}{\mathbb{R}}

\newcommand{\sign}{\mathop{\mathrm{sign}}}

\newcommand{\Ber}{\mathrm{Bernoulli}}

\newcommand{\Unif}{\mathrm{Uniform}}

\newcommand{\m}{\mathcal}

\def\dperp{\perp\!\!\!\perp}

\begin{document}

\title{Information-Distilling Quantizers}
\author{Alankrita Bhatt, Bobak Nazer, Or Ordentlich and Yury Polyanskiy
	\thanks{}
	\thanks{A. Bhatt is with the University of California, San Diego, CA, USA (email: a2bhatt@eng.ucsd.edu). B. Nazer is with Boston Univerity, Boston, MA, USA (email: bobak@bu.edu)
		O. Ordentlich is with the School of Computer Science and Engineering, Hebrew University of Jerusalem, Israel (email: or.ordentlich@mail.huji.ac.il).  Y. Polyanskiy is with the Massachusetts Institute of Technology, MA, USA (email: yp@mit.edu)}
	\thanks{This work was supported, in part, by ISF under Grant 1791/17, by the NSF CAREER award CCF-12-53205, the Center for Science of Information (CSoI), an NSF Science and Technology Center, under grant agreement CCF-09-39370, and NSF grants CCF-1618800, CCF-17-17842, and ECCS-1808692.
	The material in this paper was presented in part at the 2017 International Symposium on Information Theory~\cite{nop17}.}
	\thanks{}}

\maketitle

\parskip 3pt

\maketitle
\begin{abstract} Let $X$ and $Y$ be dependent random variables. This paper considers the problem of designing a scalar quantizer for $Y$ to maximize the mutual information between the quantizer's output and $X$, and develops fundamental properties and bounds for this form of quantization, which is connected to the log-loss distortion criterion. The main focus is the regime of low $I(X;Y)$, where it is shown that, if $X$ is binary, a constant fraction of the mutual information can always be preserved using $\mathcal{O}(\log(1/I(X;Y)))$ quantization levels, and there exist distributions for which this many quantization levels are necessary. Furthermore, for larger finite alphabets $2 < |\mathcal{X}| < \infty$, it is established that an $\eta$-fraction of the mutual information can be preserved using roughly $(\log(| \mathcal{X} | /I(X;Y)))^{\eta\cdot(|\mathcal{X}| - 1)}$ quantization levels.
\end{abstract}

\section{Introduction}\label{sec:intro}


Let $X$ and $Y$ be a pair of random variables with alphabets $\m{X}$ and $\m{Y}$, respectively, and a given distribution $P_{XY}$. This paper deals with the problem of quantizing $Y$ into $M<|\m{Y}|$ values, under the objective of maximizing the mutual information between the quantizer's output and $X$. With a slight abuse of notation\footnote{This notation is meant to suggest the distance from a point to a set.}, we will denote the value of the mutual information attained by the optimal $M$-ary quantizer by
\begin{align}
I(X;[Y]_M) \triangleq \sup_{\tilde{Y} \in [Y]_M} I(X; \tilde{Y}) .
\label{eq:gmdef}
\end{align}  where $[Y]_M$ is the set of all (deterministic) $M$-ary quantizations of $Y$,
 $$[Y]_M \triangleq \{f(Y) \  | \  f:\m{Y}\rightarrow [M]\}$$ and $[M] \triangleq \{1,2,\ldots,M\}$.

%
%
%
%

When $X$ and $Y$ are thought of as the input and output of a channel, this problem corresponds to determining the highest available information rate for $M$-level quantization. It is therefore not surprising that this problem has received considerable attention. For example, it is well known~\cite[Section 2.11]{vo13} that when $X$ is $\pm 1$ equiprobable and $Y= X+Z$ for Gaussian $Z$, statistically independent of $X$, it holds that $I(X;[Y]_2) \geq \frac{2}{\pi} I(X;Y)$, which is achieved by taking $f(\cdot)$ to be the maximum a posteriori (MAP) estimator of $X$ from $Y$.\footnote{In~\cite{kl13}, it was demonstrated that if an asymmetric signaling scheme is used, instead of binary phase-shift keying (BPSK), the additive white Gaussian noise channel capacity can be attained at low signal-to-noise ratio (SNR) with an asymmetric $2$-level quantizer.}

A characterization of~\eqref{eq:gmdef} is also required for the construction of good polar codes~\cite{arikan09}, since the large output cardinality of polarized channels makes it challenging to evaluate their respective capacities (and identify ``frozen'' bits). Efficient techniques for channel output quantization that preserve mutual information have been developed to overcome this obstacle, and played a major role in the process of making polar codes implementable~\cite{phtt11,tsv12,kt17}. One byproduct of these efforts is a sharp characterization of the \textit{additive gap}. Specifically, it was recently shown in~\cite{kt17} that, for arbitrary $P_{XY}$, it holds that $I(X;Y)-I(X;[Y]_M)= \m{O}(M^{- 2/(|\m{X}|-1)})$, whereas~\cite{tal15} demonstrates that there exist $P_{XY}$ such that $I(X;Y)-I(X;[Y]_M) = \Omega(M^{- 2/(|\m{X}|-1)})$. The works~\cite{phtt11,tsv12,kt17}, among others, also provided polynomial-complexity, sub-optimal algorithms for designing such quantizers. In addition, for binary $X$, an algorithm for determining the optimal quantizer was proposed in~\cite{ky14} (drawing upon a result from~\cite{bpkn92}) that runs in time $\m{O}(|\m{Y}|^3)$. 
A supervised learning algorithm, for the scenario where $P_{XY}$ is not known, and cannot be estimated with good accuracy, was proposed in~\cite{lr09}.

It may at first appear surprising that the quality of quantization found in~\cite{phtt11,tsv12,kt17} depends on the alphabet size $|\m{X}|$ but not on $| \m{Y}|$. The reason for this is that, given $Y=y$, the relevant information about $X$ is the the posterior distribution $P_{X|Y=y}$, which is a point on ($|\m{X}|-1$)-dimensional simplex. Thus, the goal of quantizing $Y$ is essentially a goal of quantizing the probability simplex. The goal of this paper is to understand the fundamental limits of this quantization, as a function of alphabet size. The crucial difference with~\cite{phtt11,tsv12,kt17} is that here we focus on the \textit{multiplicative gap}, i.e., comparing the ratio of $I(X;[Y]_M)$ to $I(X;Y)$. The difference is especially profound in the case when $I(X;Y)$ is small. We ignore the algorithmic aspects of finding the optimal $M$-level quantizer and instead focus on the fundamental properties of the function $I(X;[Y]_M)$. To this end, we define and study the ``information distillation'' function
\begin{align}
\mathrm{ID}_M(K,\beta) \triangleq \inf_{\substack{ {P_{XY}:}\\{|\m{X}|=K}\\{I(X;Y)\geq\beta}}} I(X;[Y]_M).\label{eq:gmbetadef}
\end{align}
The infimum above is taken with respect to all joint distributions with discrete input alphabet $\m{X}$ of cardinality $K$ and \emph{arbitrary} (possibly continuous) output alphabet $\mathcal{Y}$ such that the mutual information is at least $\beta$. Corollary~\ref{cor:infineq} stated and proved in Section~\ref{sec:properties} shows that taking the infimum in~\eqref{eq:gmbetadef} with respect to joint distributions with mutual information $I(X;Y)=\beta$, rather than $I(X;Y)\geq\beta$ would lead to the same result.  One may further wonder whether $K$ has an essential role in the function $\mathrm{ID}_M(K,\beta)$. Proposition~\ref{prop:SDPI}, stated and proved in Section~\ref{sec:properties}, shows that for any $M$ and $\beta$ it holds that $\inf_{K}\mathrm{ID}_M(K,\beta)=0$. Thus, one must indeed restrict the cardinality of $\m{X}$ in~\eqref{eq:gmbetadef} in order to get a meaningful quantity.

Special attention will be given to the binary input alphabet case, where $X\sim \Ber(p)$ for some $p$. In this setting, it may seem at a first glance that the optimal binary quantizer should always retain a significant fraction of $I(X;Y)$, and that the MAP quantizer should be sufficient to this end. For large $I(X;Y)$, this is indeed the case, as we show in Proposition~\ref{prop:highSNR}. As mentioned above, this is also the case if $Y = X + Z$ with $Z$ Gaussian for all values  $I(X;Y)$, since the MAP quantizer always retains at least $2/\pi\approx 63.66 \%$ of the mutual information. However, perhaps surprisingly, we show that there is no constant $c>0$ such that $I(X;[Y]_2)>c\cdot I(X;Y)$ for all $P_{XY}$ with $|\m{X}|=2$.

\subsection{Main Results}

Our main result is a complete characterization, up to constants, of the binary information distillation function. 

\begin{theorem}
For any mutual information value $0<\beta\leq1$, the binary information distillation function is lower and upper bounded as follows:
\begin{align}
\beta \cdot f_{\text{lower}}\left(\frac{M-1}{\max\left\{\log\left(\frac{1}{\beta}\right),1\right\}}\right) \leq \mathrm{ID}_M(2,\beta)~~~~~~~~~~~~~~~ \nonumber\\
~~~~~~~~~~~~~~~~~ \leq \beta \cdot f_{\text{upper}}\left(\frac{M-1}{\max\left\{\log\left(\frac{1}{\beta}\right),1\right\}}\right) \, ,\nonumber
\end{align}
where
\begin{align}
f_{\text{lower}} (t)\triangleq\begin{cases}
\frac{t}{208} & t<104\\
1-\frac{52}{t} & t\geq 104
\end{cases}, \qquad \ \ \ \ \ \ \ \ \ 
f_{\text{upper}}(t)\triangleq \min\{3t,1\}. \nonumber
\end{align}
\label{thm:bindistillation}
\end{theorem}

The proof is deferred to Section~\ref{sec:proof}. Note that the negative aspect of this result is in stark contrast to the intuition from the binary additive white Gaussian noise (AWGN) channel. While for the former, two quantization levels suffice for retaining a $2/\pi$ fraction of $I(X;Y)$, Theorem~\ref{thm:main} shows that there exist sequences of distributions for which at least $\Omega(\log(1/I(X;Y)))$ quantization levels are needed in order to retain a fixed fraction of $I(X;Y)$. Furthermore, as illustrated in Section~\ref{sec:binary}, for small $I(X;Y)$ and $M =2$, the MAP quantizer can be arbitrarily bad with respect to the optimal quantizer, which is in general not ``symmetric.'' On the positive side, $\m{O}(\log(1/I(X;Y)))$ quantization levels always suffice for retaining a fixed fraction of $I(X;Y)$.

For the general case where $2 < |\m{X}| < \infty$, we prove the following.

\begin{theorem}\label{thm:nonbindistillation}
	Define
	\begin{align}
	&a_0(M,|\m{X}|,\beta)\triangleq\frac{1}{|\m{X}|-1}\cdot\min\left\{\frac{M-1}{208\log\left(\frac{(|\m{X}|-1)^2}{\beta}\right)},\frac{1}{2} \right\},\\
	&a_{|\m{X}|-1}(M,|\m{X}|,\beta)\triangleq \left(1-\left(\frac{52\log\left(\frac{e(|\m{X}|-1)}{\beta}\right)}{M^{1/{(|\m{X}|-1)}}}\right)^{2/3}\right)^2,
	\end{align}
	and, for $k=1,\ldots,|\m{X}|-2$, define 
	\begin{align}
	a_k(M,|\m{X}|,\beta)&\triangleq\frac{k}{|\m{X}|-1}\left(1-\frac{52\log\left(\frac{(|\m{X}|-1)^2}{\beta}\right)}{M^{1/k}}\right).
	\end{align}
	Then, for any $2< |\m{X}| < \infty$ and $0<\beta\leq \log|\m{X}|$, the information distillation function is upper and lower bounded as follows 	
	\begin{align}
	\beta\cdot\max_{k\in\{0,1,\ldots,|\m{X}|-1\}}&a_k\left(M,|\m{X}|,\beta\right)\leq\mathrm{ID}_M(|\m{X}|,\beta)\nonumber\\
	&\leq \beta \cdot f_{\text{upper}}\left(\frac{M-1}{\max\left\{\log\left(\frac{1}{\beta}\right),1\right\}}\right),
	\end{align}
	where $f_{\text{upper}}(t)$ is as defined in Theorem~\ref{thm:bindistillation}.
%
\end{theorem}
The proof of the lower bound is deferred to Section~\ref{sec:nonbinary}, whereas the upper bound follows trivially by noting that $|\m{X}|\mapsto \mathrm{ID}_M(|\m{X}|,\beta)$ is monotone non-increasing and invoking the upper bound on $\mathrm{ID}_M(2,\beta)$ from Theorem~\ref{thm:bindistillation}. 

The lower bound from Theorem~\ref{thm:nonbindistillation} states that for any $\epsilon>0$, all $P_{XY}$ and $k\in[|\m{X}|-1]$, it holds that $M=\m{O}\left(\left(\frac{1}{\epsilon\log(| \m{X} | /I(X;Y))}\right)^{k}\right)$ suffices to guarantee that $I(X;[Y]_M)>\frac{k}{|\m{X}|-1}(1-\epsilon)I(X;Y)$. In particular, choosing $k=1$, we obtain that $M=\m{O}\left(\frac{1}{\epsilon\log(| \m{X} |/I(X;Y))}\right)$ suffices to attain $I(X;[Y]_M)>\frac{1}{|\m{X}|-1}(1-\epsilon)I(X;Y)$ and, on the other hand, by the upper bound, there exist $P_{XY}$ for which $M=\Omega(1/\log(1/I(X;Y)))$ is required in order to attain $I(X;[Y]_M)>\frac{1}{|\m{X}|-1}I(X;Y)$. Thus, Theorem~\ref{thm:nonbindistillation} gives a tight characterization (up to constants independent of $I(X;Y)$) of the number of quantization levels required in order to maintain a fraction of $0<\eta<\frac{1}{|\m{X}|-1}$ of $I(X;Y)$. However, we were not successful in establishing an upper bound that match the lower bound within the range $\eta\in\left(\frac{1}{|\m{X}|-1},1\right)$. We nevertheless conjecture that for $\eta$ close to $1$ our lower bound is tight.

\begin{conjecture}
For any $|\m{X}|>2$, there exists some $\frac{|\m{X}| - 2}{|\m{X}|-1}
<\eta(|\m{X}|)<1$, $\beta(|\m{X}|)>0$ and a constant $c(|\m{X}|)>0$, such that for all $0<\beta<\beta(|\m{X}|)$ and $M<c(|\m{X}|)(\log(1/\beta))^{|\m{X}|-1}$, it holds that
\begin{align}
\mathrm{ID}_M(|\m{X}|,\beta)<\eta(|\m{X}|)\cdot\beta.
\end{align}
\end{conjecture}


As discussed above, prior work~\cite{phtt11,tsv12,kt17} has focused on bounding the additive gap. This corresponds to bounding the so-called ``degrading cost''~\cite{tal15,kt17}, which is defined as
\begin{align}
\mathrm{DC}(|\m{X}|,M) \triangleq \sup_{0<\beta\leq\log |\m{X}|} \beta - \mathrm{ID}_M(|\m{X}|,\beta)
\end{align}  
in our notation. In particular, the bound derived in~\cite{kt17} on $\mathrm{DC}(|\m{X}|,M)$ is equivalent to the following ``constant-gap'' result: for every $0<\beta\leq \log{|\m{X}|}$, $$ \mathrm{ID}_M(|\m{X}|,\beta) \geq \beta -  \nu(|\m{X}|) M^{- 2/(|\m{X}|-1)}$$ for some function $\nu$.\footnote{It is also demonstrated in~\cite{tal15} that there exist values of $\beta$, for which this bound is tight. Specifically,~\cite{tal15} found a distribution $P_{XY}$ with $X\sim\Ber(1/2)$ and $I(X;Y)\approx0.2787$ for which $I(X;[Y]_M)<I(X;Y)-c M^{- 2}$ for some constant $c>0$.} For small $\beta$, however, results of this form are less informative. Indeed, for small $\beta$, this bound requires $M$ to scale like $\beta^{-(|\m{X}|-1)/2}$ in order to preserve a constant fraction of the mutual information. On the other hand, our result shows that scaling $M$ like $\m{O}((\log(1/\beta))^{|\m{X}|-1})$ suffices for joint distributions $P_{XY}$.

\underline{Notation:} In this paper, logarithms are generally taken ro base $2$, and all information measures are given in bits.  When a logarithm is taken to base $e$, we use the notation $\ln$ instead of $\log$. We denote the binary entropy function by $h(t)=-t\log(t)-(1-t)\log(1-t)$, and its inverse restricted to the interval $[0,1/2]$ by $h^{-1}(t)$. The notation $\lfloor t \rfloor$ denotes the ``floor'' operation, i.e., the largest integer smaller than or equal to $t$.

\section{Properties of $I(X;[Y]_M)$}\label{sec:properties}

Let $P_{XY}$ be a joint distribution on $\m{X}\times\m{Y}$ and consider the function $I(X;[Y]_M)$, as defined in~\eqref{eq:gmdef}. The restriction to deterministic functions incurs no loss of generality, see e.g.,~\cite{ky14}. Indeed, any random function of $y$, can be expressed as $f(y,U)$ where $U$ is some random variable statistically independent of $(X,Y)$. Thus,
\begin{align}
I(X;f(Y,U))\leq I(X;f(Y,U),U)=I(X;f(Y,U)|U)\label{eq:detran}
\end{align}
and hence there must exist some $u$ for which $I(X;f(Y,u))\geq I(X;f(Y,U))$. 
Furthermore, for any function $f: \m{Y} \rightarrow [M]$, we can associate a disjoint partition of the $|\m{X}|$-dimensional cube $[0,1]^{|\m{X}|}$ into $M$ regions $\m{I}_1,\ldots,\m{I}_M$, such that $f(y)=i$ iff $P_{X|Y=y}\in\m{I}_i$ for $i=1,\ldots,M$.\footnote{This statement may fail if there exist $y_0\neq y_1\in\m{Y}$ with $P_{X|Y=y_0}=P_{X|Y=y_1}$, as the function $f$ may map them to different values $f(y_0)\neq f(y_1)$. However, we can introduce the random variable $T$ which merges all values of $y\in\m{Y}$ with the same $P_{X|Y=y}$ into a single value. By Fisher's factorization criterion~\cite[Section 3.1]{pwLectureNotes}, $T$ is a sufficient statistic of $Y$ for $X$. By Proposition~\ref{prop:dpi} below, we therefore have that $I(X;[Y]_M)=I(X;[T]_M)$. Therefore, even if there exist various symbols in $\m{Y}$ with the same $P_{X|Y=y}$, there is no loss of generality in restricting attention to quantizers associated with an $M$-partition of  $[0,1]^{|\m{X}|}$.} A result of Burshtein et al.~\cite[Theorem 1]{bpkn92} shows that the maximum in~\eqref{eq:gmdef} can without loss of generality be restricted to functions for which there exists an associated partition where the regions $\m{I}_1,\ldots,\m{I}_M$ are all convex.

Below, we state simple upper and lower bounds on $I(X;[Y]_M).$

\begin{proposition}[Simple bounds]
For any distribution $P_{XY}$ on $\m{X}\times\m{Y}$ with a finite output alphabet, and $M<|\mathcal{Y}|$,
\begin{align}
\frac{M-1}{|\mathcal{Y}|}I(X;Y)\leq I(X;[Y]_M) \leq \min\{I(X;Y),\log(M)\}.\nonumber
\end{align}
\label{prop:trivbounds}
\end{proposition}

\begin{proof}
The upper bound does not require any assumptions on $\m{Y}$ and follows from the data processing inequality ($X-Y-f(Y)$ forms a Markov chain in this order), and from $I(X;f(Y))\leq H(f(Y))\leq \log(M)$.

For the lower bound, we can identify the elements of $\mathcal{Y}$ with $\{1,\ldots,|\mathcal{Y}|\}$ such that
\begin{align}
P_Y(1)D(P_{X|Y=1}||P_X)\geq\cdots\geq P_Y(|\mathcal{Y}|)D(P_{X|Y=|\mathcal{Y}|}||P_X)\nonumber
\end{align}
and take the quantization function
\begin{align*}
    f(y)=\begin{cases}
    y & \text{if }y<M,\\
    M & \text{otherwise.}
    \end{cases}
\end{align*}
Since $P_Y(y)D(P_{X|Y =y}\|P_X)$ is monotonically decreasing in $y$, we have
\begin{align*}
\frac{1}{M-1}\sum_{y=1}^{M-1} P_{Y}(y) D(P_{X|Y =y}\|P_X) &\ge \frac{1}{|\mathcal{Y}|}\sum_{y \in \mathcal{Y}} P_{Y}(y) D(P_{X|Y =y}\|P_X)\nonumber\\
&=\frac{1}{|\mathcal{Y}|}I(X;Y).
\end{align*}
Therefore, it follows that 
\begin{align*}
I(X;f(Y)) &= \sum_{y=1}^{M-1}P_{Y}(y) D(P_{X|Y =y}\|P_X) +  \left(\sum_{y\ge M} P_{Y}(y)\right)\cdot D(P_{X|f(Y)=M}\|P_X) \\
&\ge \sum_{y=1}^{M-1}P_{Y}(y) D(P_{X|Y =y}\|P_X)\nonumber\\
& \ge \frac{M-1}{|\mathcal{Y}|} I(X;Y),
\end{align*}
as claimed.
\end{proof}

For $K<M$, we can construct a (possibly sub-optimal) $K$-level quantizer by first finding the optimal $M$-level quantizer and then quantizing its output to $K$-levels. This together with the lower bound in Proposition~\ref{prop:trivbounds}, yields the following.
\begin{corollary}
For natural numbers $K<M$ we have
\begin{align}
I(X;[Y]_K) \geq\frac{K-1}{M} I(X;[Y]_M).\nonumber
\end{align}
\label{cor:doublequant}
\end{corollary}

\begin{remark}
	It is tempting to expect that $I(X;[Y]_M)$ will have ``diminishing returns'' in $M$ for any $P_{XY}$, i.e., that it will satisfy the inequality $I(X;[Y]_{M_1 \cdot M_2})\leq I(X;[Y]_{M_1}) + I(X;[Y]_{M_2})$. However, as demonstrated by the following example, this is not the case. Let $X\sim\Unif(\{0,1,2,3\})$ and $Y=[X+Z]\bmod 4$, where $Z$ is additive noise statistically independent of $X$ with $\Pr(Z=0)=\delta$ and $\Pr(Z=1)=\Pr(Z=2)=\Pr(Z=3)=(1-\delta)/3$. Since $H(Z)=h(\delta)+(1-\delta)\log(3)$, we have that
	\begin{align}
	I(X;[Y]_4)=I(X;Y)=2-h(\delta)-(1-\delta)\log(3).\label{eq:remfullent}
	\end{align}
	Furthermore, by the symmetry of the distribution $P_{XY}$, it is clear that that $I(X;f(Y))$ depends on $f$ only through $|f^{-1}(1)|=|\{y\in\m{Y} \ : \ f(y)=1 \}|$. Thus, $I(X;[Y]_2)$ is attained either by the quantizer $f(0)=1$, $f(1)=f(2)=f(3)=2$, for which $I(X;f(Y))=h\left(\frac{1}{4}\right)-\frac{1}{4}h(\delta)-\frac{3}{4}h\left(\frac{1-\delta}{3}\right)$, or by the quantizer $f(0)=f(1)=1$, $f(2)=f(3)=2$, for which $I(X;f(Y))=1-h\left(\frac{1+2\delta}{3}\right)$.
Comparing the two expressions, we see that
	\begin{align}
	I(X;[Y]_2)=\begin{cases}
	h\left(\frac{1}{4}\right)-\frac{1}{4}h(\delta)-\frac{3}{4}h\left(\frac{1-\delta}{3}\right) & \delta\leq 1/4,\\
	1-h\left(\frac{1+2\delta}{3}\right) & \delta>1/4.
	\end{cases}
	\label{eq:remquantent}
	\end{align}
	It follows from straightforward computation that~\eqref{eq:remfullent} and~\eqref{eq:remquantent} imply that $ 2I(X;[Y]_2)< I(X;[Y]_4)$ for all $\delta\notin\{1/4,1\}$.
\end{remark}

\begin{proposition}[Data processing inequality]
If $X-Y-V$ form a Markov chain in this order, then
\begin{align}
I(X;[V]_M) \leq I(X;[Y]_M).\nonumber
\end{align}
\label{prop:dpi}
\end{proposition}

\begin{proof}
For any function $f:\m{V}\to[M]$ we can generate a random function $\tilde{f}:\m{Y}\to[M]$ which first passes $Y$ through the channel $P_{V|Y}$ and then applies $f$ on its output. By~\eqref{eq:detran}, we can always replace $\tilde{f}$ by some deterministic function $\bar{f}:\m{Y}\to[M]$ such that
\begin{align}
I(X;\bar{f}(Y))\geq I(X;\tilde{f}(Y))=I(X;f(V)).\nonumber
\end{align}
\end{proof}

\begin{corollary}
For any integer $K$ and $\beta>0$ we have that
\begin{align}
\inf_{\substack{ {P_{XY}:}\\{|\m{X}|=K}\\{I(X;Y)\geq\beta}}} I(X;[Y]_M)=\inf_{\substack{ {P_{XY}:}\\{|\m{X}|=K}\\{I(X;Y)=\beta}}} I(X;[Y]_M).\nonumber
\end{align}
\label{cor:infineq}
\end{corollary}

\begin{proof}
For any $P_{XY}=P_X P_{Y|X}$ on $\m{X}\times{Y}$ with $I(X;Y)>\beta$, define the joint distribution $P_{X\tilde{Y}}=P_X P_{\tilde{Y}|X}$ on $\m{X}\times\{\m{Y}\cup\{?\}\}$, where the channel $P_{\tilde{Y}|X}$ is the concatenation of the channel $P_{Y|X}$ and an erasure channel $P_{\tilde{Y}|Y}$ that outputs $\tilde{Y}=Y$ with probability $\tfrac{\beta}{I(X;Y)}$ and $\tilde{Y}=\, \,?$ with probability $1-\tfrac{\beta}{I(X;Y)}$. Clearly, $X-Y-\tilde{Y}$ form a Markov chain in this order. By Proposition~\ref{prop:dpi}, we therefore have that $I(X;[\tilde{Y}]_M)\leq I(X;[Y]_M)$. In addition, we have that $I(X;\tilde{Y})=\beta$, and consequently,
\begin{align}
\inf_{\substack{ {P_{XY}:}\\{|\m{X}|=K}\\{I(X;Y)=\beta}}} I(X;[Y]_M)\leq\inf_{\substack{ {P_{XY}:}\\{|\m{X}|=K}\\{I(X;Y)\leq\beta}}} I(X;[Y]_M).\nonumber
\end{align}
The inequality in the other direction is obvious.
\end{proof}

\begin{proposition}
For a fixed $P_X$, the function $P_{Y|X}\mapsto I(X;[Y]_M)$ is convex.
\end{proposition}

\begin{proof}
For any $f:\m{Y}\to[M]$, let $I^f(P_X \times P_{Y|X})\triangleq I(X;f(Y))$, and note that
\begin{align}
I(X;[Y]_M)=\sup_{f:\m{Y}\to[M] } I^f(P_{X}\times P_{Y|X}).\nonumber
\end{align}
Since the supremum of convex functions is also convex, it suffices to show that for a fixed $P_X$ the function $I^f(P_{X}\times P_{Y|X})$ is convex in $P_{Y|X}$.
To this end, consider two channels $P^1_{Y|X}$ and $P^2_{Y|X}$, and let $P^1_{f(Y)|X}$ and $P^2_{f(Y)|X}$, respectively, be the induced channels from $X$ to $f(Y)$. Clearly, for the channel $\alpha P^1_{Y|X}+(1-\alpha)P^2_{Y|X}$, the induced channel is $\alpha P^1_{f(Y)|X}+(1-\alpha)P^2_{f(Y)|X}$. Let $Z\in[M]$ be the output of this channel, when the input is $X$. From the convexity of the mutual information with respect to the channel we have
\begin{align}
I^f&\left(P_X \times \left(\alpha P^1_{Y|X}+(1-\alpha)P^2_{Y|X}\right)\right)=I(X;Z)\nonumber\\
&\leq \alpha I^f (P_X \times P^1_{Y|X})+(1-\alpha)I^f (P_X\times P^2_{Y|X}),\nonumber
\end{align}
as desired.
\end{proof}

\begin{remark}
In contrast to mutual information, the functional $I(X;[Y]_M)$ is in general not concave in $P_X$ for a fixed $P_{Y|X}$. To see this consider the following example: $\m{X}=\m{Y}=\{1,2,3\}$, $M=2$, and the channel from $X$ to $Y$ is clean, i.e., $Y=X$. Let $P_{X_1}=(\tfrac{1}{2},\tfrac{1}{4},\tfrac{1}{4})$ and $P_{X_2}=(\tfrac{1}{4},\tfrac{1}{4},\tfrac{1}{2})$. Clearly, $I(X_1;[Y]_M) =1$ is attained by the quantizer $f(1)=1$, $f(2)=f(3)=2$, and $I(X_2;[Y]_M)=1$ is attained by the quantizer $f(1)=f(2)=1$, $f(3)=2$. For any $\alpha\in(0,1)$, let $P_{X} = \alpha P_{X_1} + (1-\alpha) P_{X_2}=\left(\tfrac{1+\alpha}{4},\tfrac{1}{4},\tfrac{2-\alpha}{4}\right)$. Since there do not exist two elements in the vector $P_Y=P_{X} =\left(\tfrac{1+\alpha}{4},\tfrac{1}{4},\tfrac{2-\alpha}{4}\right)$ that sum up to $1/2$, there does not exist any $2$-level quantizer for which $f(Y)\sim\Ber(1/2)$, and consequently
\begin{align}
I(X;[Y]_M)<1.\nonumber
\end{align}
\end{remark}

\begin{remark}[Complexity of finding the optimal quantizer]
For the special case where $Y=X$, the function $I(X;[Y]_M)$ reduces to\footnote{Recent work by Cicalese, Gargano and Vaccaro~\cite{cgv18} provides closed-form upper and lower bounds on $H([Y]_M)$.}
\begin{align}
H([Y]_M) \triangleq\sup_{\tilde{Y} \in [Y]_M}H(\tilde{Y}).\label{eq:gmx}
\end{align}
Furthermore, when $M=2$ the optimization problem in~\eqref{eq:gmx} is equivalent to
\begin{align}
\max_{\m{A}\subseteq \m{X} }\sum_{x \in \m{A}}p_x \text{ subject to: } \sum_{x \in \m{A}}p_x\leq \frac{1}{2},\label{eq:ssp}
\end{align}
where $p_x\triangleq\Pr(X=x)$, $x \in \m{X}$. The problem~\eqref{eq:ssp} is known as the \emph{subset sum problem} and is NP-hard~\cite{kpp04}. See also~\cite{cgv18}. Thus, when $|\m{X}|$ is not constrained, the problem of finding the optimal quantizer of $Y$ is in general NP-hard. Nevertheless, for the case where $\m{X}$ is binary, a dynamic programming algorithm finds the optimal $M$-level quantizer with complexity $\m{O}(M\cdot(|\m{Y}|-M))$, see~\cite{ky14,io14}.
\end{remark}

\begin{proposition}\label{prop:SDPI}
For any $\beta>0$, any natural $M$, and $n>\beta$, we have that
\begin{align}
\mathrm{ID}_M(2^n,\beta)\leq \frac{\log(M)}{n\cdot 2\log (e)}\beta.\label{eq:sdpibound}
\end{align}
Consequently, for any $\beta>0$ and natural $M$ we have that $\inf_{K}\mathrm{ID}_M(K,\beta)=0$, which motivates the restriction to finite input alphabets in our main theorems.
\end{proposition}

\begin{proof}
Let $n>\beta$ be a natural number, let $\delta=h^{-1}(1-\tfrac{\beta}{n})$, and let $Y\sim\Ber(1/2)$, $Z\sim\Ber(\delta)$, $Y\dperp Z$, and $X=Y\oplus Z$. Let $(X^n,Y^n)\sim P_{XY}^{\otimes n}$,  such that
\begin{align}
I(X^n;Y^n)=n(1-h(\delta))=\beta.
\end{align}
For product distributions $P^{\otimes n}_{XY}$ we have that for any $U$ satisfying the Markov chain $U-Y^n-X^n$, it holds that~\cite{agkn13,pw15}
\begin{align}
\frac{I(U;X^n)}{I(U;Y^n)}\leq \sup \frac{I(U;X)}{I(U;Y)},
\end{align}
where the supremum is taken with respect to all Markov chains $U-Y-X$ with fixed $P_{XY}$ and $I(U;Y)>0$. For the doubly symmetric binary source $P_{XY}$ of interest, this supremum is $(1-2\delta)^2$~\cite{pw15}, and consequently, we obtain that for any $f:\{0,1\}^n\to [M]$, it holds that
\begin{align}
I(f(Y^n);X^n)&\leq(1-2\delta)^2I(f(Y^n);Y^n)\nonumber\\
&\leq(1-2\delta)^2H(f(Y^n))\nonumber\\
&\leq (1-2\delta)^2 \log(M)\nonumber\\
&=\left(1-2h^{-1}\left(1-\frac{\beta}{n}\right)\right)^2 \log(M).\label{eq:deltaub}
\end{align}
Recalling that $h(\tfrac{1}{2}-\epsilon)=1-\sum_{k=1}^\infty \frac{\log(e)}{2k(2k-1)}(2\epsilon)^k$, we have that 
\begin{align}
h(\tfrac{1}{2}-\epsilon)\leq 1-2\log(e)\epsilon^2 \label{eq:hfirstorder}
\end{align} from which it follows that 
\begin{align}
h^{-1}(1-t)\geq \tfrac{1}{2}-\sqrt{\frac{t}{2\log(e)}}.\label{eq:invhlb}
\end{align}
Applying~\eqref{eq:invhlb} with $t=\tfrac{\beta}{n}$ in~\eqref{eq:deltaub}, yields the claimed result.
%
\end{proof}


\vspace{4mm}


\subsection{Relations to quantization for maximizing divergence}
For two distributions $P,Q$ on $\m{Y}$, $Q\ll P$, define
\begin{align}
\psi_M(P,Q)\triangleq\sup_{f:\m{Y}\to [M]}D(P^f||Q^f),
\end{align}
where $P^f$ and $Q^f$ are the distributions on $[M]$ induced by applying the function $f$ on the random variables generated by $P$ and $Q$, respectively.
A classical characterization of Gelfand-Yaglom-Perez~\cite[Section 3.4]{pwLectureNotes}, shows that $\psi_M(P,Q)\nearrow D(P\|Q)$ as $M\to \infty$. We are interested here in understanding the speed of this convergence. To this end, we prove the following result.
\begin{proposition}
For any $\beta,\epsilon>0$, there exists two distributions $P,Q$ on $\mathbb{N}$ such that $D(P\|Q)=\beta$ and $\psi_M(P,Q)\leq M\epsilon$ for any $M\in\mathbb{N}$.
\label{prop:KL}
\end{proposition}

\begin{proof}
Let $P$ and $Q$ be distributions on $[T+1]=\{1,\ldots,T+1\}$ with the following probability mass functions:
\begin{align}
P(m)&=\begin{cases}
2^{-m} & m=1,\ldots,T-1\\
2^{-(T-1)} & m=T\\
0 & m=T+1
\end{cases}\nonumber\\
Q(m)&=\begin{cases}
P(m) & 1\leq m\leq k\\
g(m)\cdot P(m) & k< m\leq T\\
1-\sum_{m=1}^{k}P(m)\\-\sum_{m=k+1}^T g(m)P(m) & m=T+1
\end{cases}\nonumber
\end{align}
where $0<g(m)\leq 1$ is some monotonically non-increasing function.
We have that
\begin{align}
D(P||Q)=\sum_{m=k+1}^{T-1} 2^{-m}\log(1/g(m))+2^{-(T-1)}\log(1/g(T)),
\end{align}
whereas for any $f:\{0,1,\ldots\}\to [M]$ we have that
\begin{align}
D(P^f||Q^f)&=\sum_{m=1}^M P(f^{-1}(m))\log\frac{P(f^{-1}(m))}{Q(f^{-1}(m))}\nonumber\\
&\leq M\cdot\max_{A\subset[T+1]} P(A)\log\frac{P(A)}{Q(A)},\label{eq:Amaximization}
\end{align}
where $f^{-1}(m)=\{y\in[T+1] \ : \ f(y)=m\}$, and the maximization in~\eqref{eq:Amaximization} is over all subsets of $[T+1]$. Furthermore, as $P(T+1)=0$, we see that the maximization in~\eqref{eq:Amaximization} may be restricted to all sets in $[T]$.
Let $A$ be a set achieving the maximum in~\eqref{eq:Amaximization}, and let $A_k\triangleq A\cap [k]$. Without loss of generality, we can assume that $A\setminus A_k\neq\emptyset$, as otherwise $P(A)\log (P(A)/Q(A))=0$. Thus, we can define $\ell\triangleq\min\{a \ : \ a\in A\setminus A_k\}$ and write
\begin{align}
P(A)&=P(A_k)+P(A\setminus A_k)\leq P(A_k)+\sum_{m=\ell}^{T}P(m) \leq P(A_k)+2\cdot2^{-\ell}\nonumber\\
Q(A)&=Q(A_k)+Q(A\setminus A_k)\geq P(A_k)+2^{-\ell}g(\ell)
\end{align}
Let $t=2^\ell P(A_k)+2$, and $\tau=2-g(\ell)$ such that the bounds above read as $P(A)\leq 2^{-\ell}t$ and $Q(A)\geq 2^{-\ell}(t-\tau)$, and
\begin{align}
P(A)\log\frac{P(A)}{Q(A)}&\leq -2^{-\ell}t\log\left(1-\frac{\tau}{t}\right).\label{eq:Dfbound}
\end{align}
We note that the function $\varphi(t)=-t\log(1-\tfrac{\tau}{t})$ is convex and monotone decreasing in the range $t>\tau$. This implies that~\eqref{eq:Dfbound} is maximized by choosing $A$ such that $P(A_k)=0$, for which $t=2$, and we obtain
\begin{align}
D(P^f\|Q^f)<M\cdot 2^{-(\ell-1)}\log\frac {2}{g(\ell)}.
\end{align}

Now, take $g(m)=2^{-\frac{\alpha 2^m}{m}}$ for some $0<\alpha\leq 1$, and note that it is indeed monotone non-increasing in $m=1,2,\ldots$,  which yields 
\begin{align}
D(P\|Q)&=\alpha\sum_{m=k+1}^{T-1}\frac{1}{m}+\frac{2\alpha}{T}=\alpha\left(\frac{2}{T}+\sum_{m=k+1}^{T-1}\frac{1}{m}\right)\label{eq:DPQ}\\
D(P^f\|Q^f)&\leq 2M\left(2^{-\ell}+\frac{\alpha}{\ell} \right)\leq M\left(2\cdot2^{-k}+\frac{2\alpha}{k} \right),\label{eq:DPfQf}
\end{align}
where we have used the fact that $\ell>k$ in the last inequality.
The statement follows by noting that we can always choose $k$ such that the right hand side of~\eqref{eq:DPfQf} is smaller than $\epsilon$ (recall that $\alpha<1$), and then we can choose $T>k$ and $\alpha$ such that the left hand side of~\eqref{eq:DPQ} is equal to $\beta$.
\end{proof}

Proposition~\ref{prop:KL} shows that for any fixed $M$, and any value of $D(P\|Q)$, the ratio $\psi_M(P,Q)/D(P\| Q)$ can be arbitrarily small.\footnote{However, under some restrictions on the distributions $P$ and $Q$, it is shown in~\cite{hbb18} that a $2$-level quantizer suffices to retain a constant fraction of $D(P\|Q)$.} Note that choosing a different $\varphi$-divergence in the definition of $\psi_M(P,Q)$ instead of the KL-divergence, could lead to very different results. In particular, under the total variation criterion, the $1$-bit quantizer $f(y)=\sign(P(y)-Q(y))$ achieves $d_{\text{TV}}(P^f,Q^f)=d_{\text{TV}}(P,Q)$ for any pair of distributions $P,Q$ on $\m{Y}$. An interesting question for future study is for which $\varphi$-divergences is the ratio $\psi_M(P,Q)/D_{\varphi}(P,Q)$ always positive.


\section{Bounds for binary $X$}
\label{sec:binary}

In this section, we consider the case of $|\m{X}|=2$, and provide upper and lower bounds on $\mathrm{ID}_M(2,\beta)$. We begin by studying the case where $M=|\m{X}|=2$, through which we shall demonstrate why the multiplicative decrease in mutual information is small when $I(X;Y)$ is high (close to $1$). These findings illustrate that the more interesting regime for $\mathrm{ID}_{M}(2,\beta)$ is the one where $\beta$ is small. For this regime, we derive lower and upper bound that match up to constants that do not depend on $\beta$.

\subsection{Binary Quantization ($M=2$)}

The aim of this subsection is to analyze the performance of quantizers whose cardinality is equal to that of $\m{X}$. In this case, a natural choice for the quantizer is the maximum a posteriori (MAP) estimator of $X$ from $Y$. Intuitively, when $I(X;Y)$ is high (close to $H(X)$), the MAP estimator should not make many errors and the mutual information between it and $X$ should be high as well. We make this intuition precise below. However, when $I(X;Y)$ is low, it turns out that not only does the MAP estimator fail to retain a significant fraction of $I(X;Y)$, but it can be significantly inferior to other binary quantizers. 

Assume without loss of generality that $\mathcal{X} = \{1,2\}$. The maximum a posteriori (MAP) quantizer is defined by
\begin{align}
    f_{\text{MAP}}(y)=\begin{cases}
    1 & \text{if }\Pr(X=1|Y=y)> 1/2\\
    2 & \text{if }\Pr(X=1|Y=y)< 1/2\\
    1\cdot U+2(1-U) & \text{if }\Pr(X=1|Y=y)= 1/2
    \end{cases},\label{eq:mapdef}
\end{align}
where $U\sim\Ber(1/2)$ is statistically independent of $(X,Y)$.
Let $P_{e,\text{MAP}}(y)\triangleq\Pr(f_{\text{MAP}}(Y)\neq X | Y=y)$ and $P_{e,\text{MAP}}\triangleq\mathbb{E}_Y P_{e,\text{MAP}}(Y)$. By the concavity of the binary entropy function $t\mapsto h(t)$, we have that $h(t)\geq 2t$ for any $0\leq t\leq 1/2$, with equality iff $t\in\{0,1/2\}$. Consequently,
\begin{align}
H(X|Y)=\mathbb{E}_Y h(P_{e,\text{MAP}}(Y))\geq 2 P_{e,\text{MAP}}.\label{eq:entpemapbound}
\end{align}
Let $X\sim\Ber(p)$ and $I(X;Y)=\beta$. We have that
\begin{align}
&I(X;f_{\text{MAP}}(Y))\\
&=H(X)-H(X|f_{\text{MAP}}(Y))\nonumber\\
&=h(p)- \Pr( f_{\text{MAP}}(Y) = 1) h\big(\Pr(X \neq 1 | f_{\text{MAP}}(Y) = 1)\big) \nonumber \\
&~~~~~~~~~~- \Pr( f_{\text{MAP}}(Y) = 2) h\big(\Pr(X \neq 2 | f_{\text{MAP}}(Y) = 2)\big) \nonumber\\
&\geq h(p)-h(P_{e,\text{MAP}})\label{eq:mapent}\\
&\geq h(p)-h\left(\frac{H(X|Y)}{2} \right)\label{eq:entpemap}\\
&=h(p)-h\left(\frac{h(p)-\beta}{2} \right),\label{eq:MAPmi}
\end{align}
where~\eqref{eq:mapent} follows from Jensen's inequality and concavity of $t\mapsto h(t)$, whereas~\eqref{eq:entpemap} follows from~\eqref{eq:entpemapbound}.
Since $\beta\leq h(p)\leq 1$, we have obtained that
\begin{align}
I(X;&f_{\text{MAP}}(Y))\geq \min_{\beta\leq t\leq 1} t-h\left(\frac{t-\beta}{2}\right)\nonumber\\
&=\begin{cases}
\beta+\frac{2}{5}-h\left(\frac{1}{5}\right) & \beta<\frac{3}{5}\\
1-h\left(\frac{1-\beta}{2} \right)  & \beta\geq\frac{3}{5}
\end{cases}.\label{eq:mapID}
\end{align}
Here,~\eqref{eq:mapID} follows by using the fact that the function $t \mapsto t - h\left(\frac{t-\beta}{2}\right)$ is convex and restricted to $\beta \le t \le 1$ to show that it attains its minimum at $t = \beta + \frac{2}{5}$ if $\beta < \frac{3}{5}$ and at $ t = 1$ otherwise. Since $I(X;[Y]_2)\geq I(X;f_{\text{MAP}}(Y))$, it follows that the right hand side of~\eqref{eq:mapID} is a lower bound on $\mathrm{ID}_2(2,\beta)$. 

In order to obtain an upper bound on $\mathrm{ID}_2(2,\beta)$, assume $X\sim\Ber(1/2)$ and $P_{Y|X}$ is the binary erasure channel (BEC), i.e., $\m{Y}=\{0,1,?\}$ and
\begin{align}
\Pr(Y=y|X=x)=\begin{cases}
\beta & \text{if } y=x\\
1-\beta & \text{if }y=?
\end{cases},\label{eq:BEC}
\end{align}
such that $I(X;Y)=\beta$. Consider the quantizer
\begin{align}
f_Z(y)=\begin{cases}
1 & \text{if } y\in\{1,?\}\\
2 & \text{if } y=0
\end{cases}.\nonumber
\end{align}
Since there exists an optimal deterministic quantizer, and any deterministic $1$-bit quantizer for the BEC output is of the form $f_Z(y)$, this must be an optimal $1$-bit quantizer. Note that the induced channel from $X$ to $f_Z(Y)$ is a $Z$-channel, and it satisfies
\begin{align}
I(X;f_Z(Y))&=1 - \left(1 - \beta + \frac{\beta}{2}\right)H(X|f_Z(Y) = 1) - \frac{\beta}{2}H(X|f_Z(Y) = 2)\nonumber\\
&=\frac{\beta}{2}h\left(\frac{1-\beta}{2-\beta}\right)+1-h\left(\frac{1-\beta}{2-\beta}\right).\label{eq:fZ}
\end{align}
By the optimality of the quantizer $f_Z(\cdot)$ for this particular distribution, it follows that the right hand side of~\eqref{eq:fZ} constitutes an upper bound on $\mathrm{ID}_2(2,\beta)$.

We have therefore established the following proposition.
\begin{proposition}
For all $3/5\leq\beta\leq 1$ we have
\begin{align}
1-h\left(\frac{1-\beta}{2}\right)\leq \mathrm{ID}_2(2,\beta)\leq 1-\frac{2-\beta}{2}h\left(\frac{1-\beta}{2-\beta}\right).
\end{align}
\label{prop:highSNR}
\end{proposition}

Thus, for large $\beta$, the loss for quantizing the output to one bit is small and the fraction of the mutual information that can be retained approaches $1$ as the mutual information increases. In particular, the natural MAP quantizer is never too bad, and retains a significant fraction of at least $1-h((1-\beta)/2)$ of the mutual information $\beta$.

In the small $\beta$ regime, we arrive at qualitatively different behavior.
We next show that the MAP quantizer can be highly sub-optimal when $\beta$ is small. To that end, consider again the distribution $X\sim\Ber(1/2)$ and $P_{Y|X}$ given by~\eqref{eq:BEC}. i.e., a BEC. It is easy to verify that in this case both inequalities in~\eqref{eq:MAPmi} are in fact equalities for all $0\leq\beta\leq 1$. It follows that for a BEC with capacity $\beta\ll 1$ and uniform input, we have that
\begin{align}
&I(X;f_{\text{MAP}}(Y))=1-h\left(\frac{1-\beta}{2}\right)=\frac{\log e}{2}\beta^2+o(\beta^2).\label{eq:mapsmall}\\
&I(X;f_Z(Y))=\frac{\beta}{2}h\left(\frac{1-\beta}{2-\beta}\right)+1-h\left(\frac{1-\beta}{2-\beta}\right)=\frac{\beta}{2}+o(\beta).\label{eq:Zsmall}
\end{align}
Thus, the asymmetric quantizer $f_Z(y)$ retains $50\%$ of the mutual information, whereas the fraction of mutual information retained by the symmetric MAP quantizer vanishes as $\beta$ goes to zero.

One can argue that $f_Z(y)$ is a MAP estimator just as $f_{\text{MAP}}(y)$, as the two quantizers attain the same error probability in guessing the value of $X$ based on $Y$, and dismiss our findings about the sub-optimality of $f_{\text{MAP}}(y)$ by attributing it to the randomness required by the MAP quantizer, as defined in~\eqref{eq:mapdef}, in the BEC setting. This is not the case however. To see this consider a channel with binary symmetric input and output alphabet $\mathcal{Y}=\{0,1\}\times\{g,b\}$, defined by
\begin{align}
\Pr(Y=y|X=x)=\begin{cases}
\beta & \text{if } y=(x,g)\\
(1-\beta)\left(\frac{1}{2}+\delta\right) & \text{if }y=(x,b)\\
(1-\beta)\left(\frac{1}{2}-\delta\right) & \text{if }y=(1-x,b)\\
\end{cases},\nonumber
\end{align}
for some $0\leq\beta\leq 1$ and $0\leq \delta\leq 1/2$. Note that for $\delta=0$, this channel becomes a BEC with capacity $1-\beta$. For any $\delta>0$, the corresponding MAP quantizer is deterministic, but as $\delta\to 0$, the channel approaches a BEC, and its performance becomes closer and closer to~\eqref{eq:mapsmall}. Similarly, the performance of a binary quantizer that assigns the same value to both ``bad'' outputs, i.e., $f(y)=2$ if $y=(0,g)$ and $f(y)=1$ otherwise, approach~\eqref{eq:Zsmall} as $\delta\to 0$.


\subsection{Lower Bound on Quantized Mutual Information}

We prove the following lower bound on $I(X;[Y]_M)$.

\begin{theorem} 	For any $P_{XY}$ with $|\m{X}|=2$ and $I(X;Y)=\beta$, and any $\eta\in (0,1)$ we have that
	\begin{align}
	I(X;[Y]_{\bar{M}_2(\eta,\beta)})\geq \eta\beta,
	\label{eq:gMlower}
	\end{align}
	where
	\begin{align}
	\bar{M}_2(\eta,\beta)\triangleq\left\lfloor c_1(\eta)\max\left\{\log\left(\frac{1}{\beta}\right),1\right\}\right\rfloor,
	\end{align}
	and
	\begin{align}
	c_1(\eta) \triangleq \frac{52}{1-\eta}.\label{eq:c1}
	\end{align}
	\label{thm:binarylb}
\end{theorem}

\begin{proof}
	Consider the joint distribution $P_{XY}$, and for any $y\in\m{Y}$ define $\alpha_y\triangleq \Pr(X=1|Y=y)$, $\bar{\alpha}\triangleq\mathbb{E}(\alpha_Y)=\Pr(X=1)$ and
	\begin{align}
	D_y\triangleq D(P_{X|Y=y}\|P_X)=d\left(\alpha_y\left\|\right.\bar{\alpha}\right),
	\end{align}
	where $d(p_1\|p_2)\triangleq p_1\log(p_1/p_2)+(1-p_1)\log((1-p_1)/(1-p_2))$ is the binary KL divergence function. Let 
	\[
	\kappa \triangleq \max\left\{\log\left(\frac{1}{\bar{\alpha}}\right),\log\left(\frac{1}{1-\bar{\alpha}}\right)\right\}=\max\{d(1\|\bar{\alpha}),d(0\|\bar{\alpha})\},	\]
	and note that the convexity of $\alpha_y\mapsto d(\alpha_y\|\bar{\alpha})$ implies that $D_y\leq\kappa$ for all $y\in\m{Y}$. 
	We further define the function
	\begin{align}
	\bar{F}(\gamma)\triangleq \Pr(D_Y\geq \gamma),
	\end{align}
	and note that it is non-increasing and satisfies
	\begin{align}
	I(X;Y)=\mathbb{E}D_Y=\int_{0}^{\gamma^*}\bar{F}(\gamma)d\gamma,
	\label{eq:DYint}
	\end{align}
	where $\gamma^*=\sup_{y\in\m{Y}} D_y\leq\kappa$. Let $L$ be some natural number, let $0=\gamma_0\leq\gamma_1\leq\cdots\leq\gamma_L\leq\gamma_{L+1}=\gamma^*+\delta$, for some arbitrary small $\delta>0$, and define the following $(2L+1)$-level quantizer
	\begin{align}
	f(y)=\begin{cases}
	0 & d(\alpha_y\|\bar{\alpha})\leq \gamma_1\\
	-\ell & \alpha_y<\bar{\alpha}, \gamma_\ell\leq d(\alpha_y\|\bar{\alpha})<\gamma_{\ell+1}\\
	\ell & \alpha_y>\bar{\alpha}, \gamma_\ell\leq d(\alpha_y\|\bar{\alpha})<\gamma_{\ell+1}
	\end{cases}.\label{eq:uniquant}
	\end{align}
	The function $f$ above, induces a partition of the interval $[0,1]$ to the intervals $\m{I}_\ell=\{\alpha_y\in[0,1] \ | \ f(y)=\ell \}$, for $\ell\in\{-L,\ldots,L\}$. Consequently, $\mathbb{E}[\alpha_Y|f(Y)=\ell]\in\m{I}_\ell$, and we have that for $\ell=1,\ldots,L$
	\begin{align}
	d\left(\mathbb{E}[\alpha_Y|f(Y)=-\ell]\|\bar{\alpha}\right)\geq \gamma_{\ell}, \ d\left(\mathbb{E}[\alpha_Y|f(Y)=\ell]\|\bar{\alpha}\right)\geq \gamma_{\ell}.\nonumber
	\end{align}
	Furthermore, by the definition of $\bar{F}(\gamma)$ we also have
	\begin{align}
	\Pr\left(\{f(Y)=-\ell\}\cup \{f(Y)=\ell\} \right)=\bar{F}(\gamma_{\ell})-\bar{F}(\gamma_{\ell+1}).\nonumber
	\end{align}
	Thus,
	\begin{align}
	I(&X;f(Y))\nonumber\\&= \sum_{\ell = -L}^L \Pr(f(Y) = \ell) D(P_{X|f(Y)=\ell}\|P_X) \nonumber\\
	& \geq \sum_{\ell=1}^L \left(\bar{F}(\gamma_{\ell})-\bar{F}(\gamma_{\ell+1})\right) \gamma_\ell\nonumber\\
	&=\bar{F}({\gamma_1})\gamma_1+\sum_{\ell=2}^L\bar{F}(\gamma_\ell)(\gamma_{\ell}-\gamma_{\ell-1})-\bar{F}(\gamma_{L+1})\gamma_{L}\nonumber\\
	&=\sum_{\ell=1}^L\bar{F}(\gamma_\ell)(\gamma_{\ell}-\gamma_{\ell-1}),\label{eq:Ifxsum}
	\end{align}
	where in the last equality we used $\gamma_0=0$ and $\bar{F}(\gamma_{L+1})=\bar{F}(\gamma^*+\delta)=0$. Our goal is therefore to choose the numbers $\{\gamma_{\ell}\}_{\ell=1}^L$ such as to maximize~\eqref{eq:Ifxsum}.
	For a general $L$, this problem is difficult and we therefore resort to a possibly suboptimal choice according to the rule
	\begin{align}
	\gamma_1=\epsilon I(X;Y),\ \theta=\left(\frac{\gamma_1}{\kappa}\right)^{-\frac{1}{L}}, \  \gamma_{\ell}=\gamma_1\cdot\theta^{\ell-1},
	\label{eq:thetarule}
	\end{align}
	for $\ell=2,\ldots,L,L+1$ and some $0<\epsilon<1$ to be specified. As $\gamma_1\leq\kappa$, we have that $\theta\geq 1$.
	Note that this choice guarantees that
	\begin{align}
	\gamma_{\ell+1}-\gamma_{\ell} \leq \theta\left(\gamma_\ell-\gamma_{\ell-1} \right), \ \ell=1,\ldots,L.
	\end{align}
	This implies that
	\begin{align}
	I(X;Y)&=\int_{0}^{\kappa}\bar{F}(\gamma)d\gamma\nonumber\\
	&=\sum_{\ell=0}^L\int_{\gamma_\ell}^{\gamma_{\ell+1}}\bar{F}(\gamma)d\gamma\nonumber\\
	&\leq \sum_{\ell=0}^L(\gamma_{\ell+1}-\gamma_{\ell})\bar{F}(\gamma_{\ell})\nonumber\\
	&\leq\gamma_1+\theta\sum_{\ell=1}^L(\gamma_{\ell}-\gamma_{\ell-1})\bar{F}(\gamma_{\ell})\nonumber\\
	&\leq \gamma_1+\theta I(X;f(Y)).
	\end{align}
	Therefore,
	\begin{align}\label{eq:epstheta}
	I(X;f(Y))&\geq \frac{(1-\epsilon)}{\theta}I(X;Y). 
	\end{align}
	Substituting in
	\begin{align}
	\epsilon = \frac{1-\eta}{2},~~L= \Bigg\lceil \frac{2\log\left(\frac{2\kappa}{(1-\eta)I(X;Y)}\right)}{(1-\eta)} \Bigg\rceil,\label{eq:epsdef}
	\end{align} it follows that 
	\begin{align}
	\theta \le 2^{\frac{(1-\eta)}{2}}.\label{eq:thetauba}
	\end{align}
	Using this, from~\eqref{eq:epstheta} it follows that
	\begin{align}
	I(X;f(Y))\geq \frac{(1-\frac{(1-\eta)}{2})}{2^{\frac{(1-\eta)}{2}}}I(X;Y) \geq \eta I(X;Y) \nonumber
	\end{align}
	where the second inequality follows since $\frac{1-x}{2^x} \ge 1-2x$ for $0 \le x \le \frac{1}{2}$. This can be established by observing that for the function $g(x) = \frac{1-x}{2^x} - (1-2x)$ we have $g(0) = 0$ and $g'(x) \ge 0$ for all $0 \le x \le \frac{1}{2}$, which implies that in this interval $g(x) \ge 0$. Recall that $|f(\cdot)|=2L+1$, and note that $2\lceil x\rceil+1\leq \lfloor 4x\rfloor$ for $x>5/4$. We have therefore shown that if $ \frac{2\log\left(\frac{2\kappa}{(1-\eta)I(X;Y)}\right)}{(1-\eta)} >5/4$, then $I(X;[Y]_M)\geq \eta I(X;Y)$ for
	\begin{align}
	M =  \Bigg\lfloor \frac{8\log\left(\frac{2\kappa}{(1-\eta)I(X;Y)}\right)}{(1-\eta)} \Bigg\rfloor.\label{eq:Mderiv}
	\end{align} 
	If $ \frac{2\log\left(\frac{2\kappa}{(1-\eta)I(X;Y)}\right)}{(1-\eta)} <5/4$, then $2L+1\leq 5$, and therefore taking $M=5$ suffices to guarantee that $I(X;[Y]_M)\geq \eta I(X;Y)$. The theorem deals with $\bar{M}_2(\eta,\beta)\geq c_1(\eta)\geq 52$, so we may assume from now on, without loss of generality, that  indeed $ \frac{2\log\left(\frac{2\kappa}{(1-\eta)I(X;Y)}\right)}{(1-\eta)} >5/4$.
	
	Now consider the case where $I(X;Y)\leq \frac{1-\eta}{2}$.
	Note that  
	\begin{align}
	&\frac{8\log\left(\frac{2\kappa}{(1-\eta)I(X;Y)}\right)}{(1-\eta)}\nonumber \\
	&= \frac{8}{1-\eta}\left( \log(\kappa) + \log\left(\frac{2}{1-\eta}\right) +  \log\left(\frac{1}{I(X;Y)}\right) \right)\nonumber\\
	&\le \frac{8}{1-\eta}\left( \log(\kappa) +  2\log\left(\frac{1}{I(X;Y)}\right) \right) \label{eq:kappaM}
	\end{align}
	where~\eqref{eq:kappaM} follows since $I(X;Y) \le \frac{1-\eta}{2}$. Next, we show that $\log(\kappa) = \m{O}\left(\log\left(\frac{1}{I(X;Y)}\right)\right)$. Without loss of generality, assume $\bar{\alpha} \le (1-\bar{\alpha})$, and so $\kappa = \log\left(\frac{1}{\bar{\alpha}}\right)$. Now note that $I(X;Y) \le H(X) = h(\bar{\alpha})$. Thus, $\bar{\alpha} \ge h^{-1}(I(X;Y))$. We then have the following bound on $h^{-1}(I(X;Y))$ \cite[Theorem 2.2]{calabro2009exponential}
	\begin{align}
	h^{-1}(I(X;Y)) \ge \frac{I(X;Y)}{2\log(\frac{6}{I(X;Y)})}. 
	\end{align}
	Therefore, $\frac{1}{\bar{\alpha}} \le \frac{1}{h^{-1}(I(X;Y))} \le \frac{2}{I(X;Y)}\log\left(\frac{6}{I(X;Y)}\right)$ which in turn implies that $\log(\kappa) = \log\log(\frac{1}{\bar{\alpha}}) \le \log\log\left(\frac{2}{I(X;Y)}\right) + \log\log\log\left(\frac{6}{I(X;Y)}\right)$. Now note that $\tau \mapsto \tau \log \left(\frac{2}{\tau}\right)$ is increasing in $0 \le \tau \le \frac{1}{2}$, which implies that $\tau \log \left(\frac{2}{\tau}\right) \le 1$ and subsequently $\log \log \left(\frac{2}{\tau}\right) \le \log \left(\frac{1}{\tau}\right)$ for $0 \le \tau \le \frac{1}{2}$ . Similarly, by noting that the function $\tau \mapsto \tau \log \log \left(\frac{6}{\tau}\right)$ is increasing in $0 \le \tau \le \frac{1}{2}$, we have $\log\log\log\left(\frac{6}{\tau}\right) \le \log\left(\frac{1}{\tau}\right)$ for $0 \le \tau \le \frac{1}{2}$. Therefore, for any $I(X;Y) \le \frac{1}{2}$,
	\begin{align}\label{eq:logkappa2}
	\log(\kappa) \le \log\log\left(\frac{2}{I(X;Y)}\right) + \log\log\log\left(\frac{6}{I(X;Y)}\right) \le  2\log\left(\frac{1}{I(X;Y)}\right).
	\end{align}
	Using this in~\eqref{eq:kappaM} we get
	\begin{align}
	\frac{8\log\left(\frac{2\kappa}{(1-\eta)I(X;Y)}\right)}{(1-\eta)} \le \frac{32}{1-\eta}\log\left(\frac{1}{I(X;Y)}\right)
	\end{align}
	for any $I(X;Y) = \beta \le \frac{1-\eta}{2}$. Therefore, $I(X;[Y]_M) \geq \eta \beta$ for $M \ge \left\lfloor \frac{32}{1-\eta}\log\left(\frac{1}{\beta}\right) \right \rfloor$ whenever $\beta \le \frac{1-\eta}{2}$. 

		When $\beta \ge \frac{1-\eta}{2}$, we use a result established in~\cite[Theorem 1]{kt17} that bounds the worst-case additive gap due to quantization  for all $M\geq 2|\m{X}|$ and $|\m{Y}|>2|\m{X}|$ as
\begin{align}
\sup_{P_{XY}} I(X;Y)-I(X;[Y]_M) \le \nu(|\m{X}|) \cdot M^{-\frac{2}{|\m{X}| - 1}},\label{eq:ktbound}
\end{align}
where the supremum is with respect to all $P_{XY}$ with input alphabet of cardinality $|\m{X}|$, and output alphabet of cardinality $|\m{Y}|$ and 
\begin{align}\label{NuDefn}
	    \nu(|\mathcal{X}|) \triangleq \frac{\pi |\mathcal{X}| (|\mathcal{X}|-1)}{2\left(\sqrt{1 + \frac{1}{2(|\mathcal{X}|-1)}}-1\right)^2} \cdot \left(\frac{2 |\mathcal{X}|}{\Gamma\left(1 + \frac{|\mathcal{X}|-1}{2}\right)}\right)^{\frac{2}{|\mathcal{X}|-1}}.
\end{align}
 Noting that $\nu(2) \le 1268$, it follows that when  $|\mathcal{X}| = 2$ and $I(X;Y) = \beta$ we have 
	\begin{align}\label{eq:additivegap}
	\beta - I(X;[Y]_M) \le 1268 M^{-2}
	\end{align}
	which implies that for $\beta \ge \frac{1-\eta}{2}$,
	\begin{align}
	I(X;[Y]_M) \ge \beta\left(1-\frac{2536 M^{-2}}{1-\eta}\right) \ge \eta \beta
	\end{align}
	whenever $M \ge \left \lfloor\frac{52}{1-\eta} \right \rfloor$. 
	
	Since
	\begin{align}
	\max\left\{\frac{52}{1-\eta}, \frac{32\log\left(\frac{1}{\beta}\right)}{1-\eta} \right\}\leq \frac{52\max\left\{\log\left(\frac{1}{\beta}\right),1\right\}}{1-\eta} ,
	\end{align}   
	combining the results obtained for each case $\left(\beta \le \frac{1-\eta}{2}  \text{ and } \beta \ge \frac{1-\eta}{2}\right)$ establishes
	\begin{align*}
	I(X;[Y]_M) \ge \eta \beta \mathrm{ \quad for \quad} M \ge \left\lfloor \frac{52\max\left\{\log\left(\frac{1}{\beta}\right),1\right\}}{1-\eta} \right\rfloor, 
	\end{align*}
	as desired.
\end{proof}

\subsection{Upper Bound on Quantized Mutual Information}

\begin{theorem}\label{thm:binaryub}
For any $0<\beta\leq 1$, there exists a distribution $P_{XY}$ with $I(X;Y)\geq\beta$, for which
	\begin{align}
	I(X;[Y]_M)\leq 2M\frac{\beta}{\ln\left(\frac{e\log(e)}{2\beta}\right)},
	\label{eq:gmupper}
	\end{align}
	for every natural $M$.
	\label{thm:main}
\end{theorem}

\begin{proof}
We provide a distribution $P_{XY}$ with $I(X;Y)\geq \beta$ for which no $M$-level quantizer achieves mutual information exceeding the right hand side of~\eqref{eq:gmupper}. Let $X\sim\Ber(1/2)$ and $Y=(V=X\oplus Z_T,T)$ be the output of a binary-input memoryless output-symmetric (BMS) channel whose input is $X$, where $Z_T$ is a binary random variable with $\Pr(Z_T=1|T=t)=t$, and $(Z_T,T)$ is statistically independent of $X$. We consider a BMS channel where  $T$ is a mixed random variable in $[0,1/2)$, whose probability density function is given by
\begin{align}
f_T(t)=\begin{cases}
r\delta(t)+\frac{4r}{(1-2t)^3} & 0^-<t\leq\frac{1-\sqrt{r}}{2}\\
0 & \text{otherwise}
\end{cases}
\end{align}
for some $0<r\leq 1$,  where $\delta(t)$ is Dirac's delta function. As above, we define $\alpha_y\triangleq \Pr(X=1|Y=y)$. Recalling that $Y=(V=X\oplus Z_T,T)$, we have that
\begin{align}
\alpha_y=\alpha_{v,t}=\Pr(X=1|V=v,T=t)=\begin{cases}
t & v=0\\
1-t & v=1
\end{cases}.
\end{align}
Furthermore, as $X\sim\Ber(1/2)$ is statistically independent of $(T,Z_T)$, we also have that $V\sim\Ber(1/2)$ is statistically independent of $T$. 
Consequently,
\begin{align}
\Pr(\alpha_Y=t|T=t)=\Pr(V=0)=\Pr(V=1)=\Pr(\alpha_Y=1-t|T=t)=1/2 \ . \label{eq:alphacondprob}
\end{align} 

By~\cite[Theorem 1]{bpkn92}, the optimal quantizer partitions the interval $[0,1]$ into $M$ subintervals $\m{I}_i=[\gamma_{i-1},\gamma_i)$ for $i=1,\ldots,M-1$ and $\m{I}_M=[\gamma_{M-1},\gamma_{M}]$, where $0=\gamma_0<\gamma_1<\cdots<\gamma_M=1$, and outputs $f(y)=i$ iff $\alpha_{y}\in\m{I}_i$. We therefore have
\begin{align}
&I(X;f(Y))=\sum_{i=1}^M \Pr(\alpha_Y\in\m{I}_i)d\left(\mathbb{E}[\alpha_Y|\alpha_Y\in\m{I}_i]\, \bigg\|\, \frac{1}{2}\right)\nonumber\\
&\leq M\max_{0\leq a<b\leq 1}\Pr(a\leq\alpha_Y\leq b)d\left(\mathbb{E}[\alpha_Y|a\leq\alpha_Y\leq b]\, \bigg\|\,\frac{1}{2}\right).\nonumber
\end{align}
By the symmetry of the random variable $\alpha_Y$ around $1/2$, we can restrict the optimization to $a<1/2$ and $a<b\leq 1$.
Let $\underline{b}=\min\{b,1-b\}$ and $\bar{b}=\max\{b,1-b\}$ and define the two intervals $\m{T}_0=[a,\underline{b})$, $\m{T}_1=[\underline{b},\bar{b}]$. By the convexity of KL divergence we have that
\begin{align}
&d\left(\mathbb{E}[\alpha_Y|a\leq\alpha_Y\leq b] \, \bigg\|\, \frac{1}{2}\right)\nonumber\\
&=d\left(\sum_{i=0}^1 \Pr(\alpha_Y\in\m{T}_i|a\leq\alpha_Y\leq b)\mathbb{E}[\alpha_Y|\alpha_Y\in\m{T}_i]\, \bigg\|\,\frac{1}{2}\right)\nonumber\\
&\leq \sum_{i=0}^1 \Pr(\alpha_Y\in\m{T}_i|a\leq\alpha_Y\leq b)d\left(\mathbb{E}[\alpha_Y|\alpha_Y\in\m{T}_i] \, \bigg\|\,\frac{1}{2}\right)\nonumber\\
&=\Pr(\alpha_Y\in\m{T}_0|a\leq\alpha_Y\leq b)d\left(\mathbb{E}[\alpha_Y|a\leq\alpha_Y\leq\underline{b}] \, \bigg\|\,\frac{1}{2}\right),\nonumber
\end{align}
where in the last equation we have used the fact that $\mathbb{E}[\alpha_Y|\alpha_Y\in\m{T}_1]=1/2$, due to the symmetry of the random variable $\alpha_Y$. We have therefore obtained
\begin{align}
&I(X;f(Y))\nonumber\\
&\leq M\max_{0\leq a\leq b\leq \tfrac{1}{2}}\Pr(a\leq\alpha_Y\leq b)d\left(\mathbb{E}[\alpha_Y|a\leq\alpha_Y\leq b] \, \bigg\|\,\frac{1}{2}\right)\nonumber\\
&\overset{(i)}{=}\frac{M}{2}\max_{0\leq a\leq b\leq \tfrac{1}{2}}\Pr(a\leq T \leq b)d\left(\mathbb{E}[T|a\leq T\leq b] \, \bigg\|\,\frac{1}{2}\right).\nonumber\\
&\overset{(ii)}{=}\frac{M}{2}\max_{0\leq b\leq \tfrac{1}{2}}\Pr(0\leq T \leq b)d\left(\mathbb{E}[T|0\leq T\leq b] \, \bigg\|\,\frac{1}{2}\right)
\end{align}
where $(i)$ follows since for any interval $\m{A}\subset[0,1/2)$ we have that $\Pr(\alpha_Y\in\m{A})=\frac{1}{2} \Pr(T\in\m{A})$ and $\mathbb{E}[\alpha_Y|\alpha_Y\in\m{A}]=\mathbb{E}[T|T\in\m{A}]$ by~\eqref{eq:alphacondprob}  and $(ii)$ follows since for any choice of $0<b\leq 1/2$, both terms are individually maximized by $a=0$.
It can be verified that for any $0\leq \rho\leq \tfrac{1-\sqrt{r}}{2}$
\begin{align}
\int_{0}^\rho tf_T(t) dt=\frac{2r\rho^2}{(1-2\rho)^2}; \ \Pr(0\leq T\leq \rho)=\frac{r}{(1-2\rho)^2},\nonumber
\end{align}
and therefore $\mathbb{E}[T|0\leq T\leq b]=2b^2$, and we have that for any $M$-level quantizer
\begin{align}
I(X;f(Y))&\leq \frac{M}{2}\cdot\max_{0\leq b\leq\tfrac{1-\sqrt{r}}{2}} r\cdot\frac{1-h(2b^2)}{(1-2b)^2}\nonumber\\
&\leq M\cdot\log(e)r,
\end{align}
where the last inequality follows by noting that the function $\tfrac{1-h(2b^2)}{(1-2b)^2}$ is monotone increasing in $0<b<1/2$, and taking the limit as $b\to1/2$. It remains to relate $r$ and $I(X;Y)$. Recalling from~\eqref{eq:hfirstorder} that $h(\tfrac{1}{2}-p)\leq 1-2\log(e)p^2$, we have
\begin{align}
I(X;Y)&=1-\mathbb{E}h(T)\nonumber\\
&\geq 2\log(e)\mathbb{E}\left(\frac{1}{2}-T\right)^2.\label{eq:preT2moment}
\end{align}
To compute the expectation above, we use the definition of $f_T(t)$ and write
\begin{align}
\mathbb{E}\left(\frac{1}{2}-T\right)^2&=\int_{0^-}^{\frac{1-\sqrt{r}}{2}}\left(\frac{1}{2}-T\right)^2 f_T(t)dt\nonumber\\
&=\int_{\frac{\sqrt{r}}{2}}^{\frac{1}{2}^+}x^2 f_T\left(\frac{1}{2}-x\right)dx\nonumber\\
&=\int_{\frac{\sqrt{r}}{2}}^{\frac{1}{2}^-}x^2 f_T\left(\frac{1}{2}-x\right)dx+\int_{\frac{1}{2}^-}^{\frac{1}{2}^+}x^2 f_T\left(\frac{1}{2}-x\right)dx\nonumber\\
&=\int_{\frac{\sqrt{r}}{2}}^{\frac{1}{2}^-}\frac{4rx^2}{(2x)^3}dx +\int_{\frac{1}{2}^-}^{\frac{1}{2}^+}x^2\cdot r\delta\left(\frac{1}{2}-x\right)dx\nonumber\\
&=\frac{r}{2}\ln\left(\frac{1/2}{\sqrt{r}/2}\right)+\frac{r}{4}\nonumber\\
&=\frac{r}{4}\ln\left(\frac{e}{r}\right).\label{eq:T2moment}
\end{align}
Substituting~\eqref{eq:T2moment} into~\eqref{eq:preT2moment}, yields
\begin{align}
I(X;Y)&= 2\log(e)\frac{r}{4}\ln\left(\frac{e}{r} \right)\nonumber\\
&= \frac{e\log(e)}{2}\frac{r}{e}\ln\left(\frac{e}{r} \right)\nonumber.
\end{align}
It can be verified that the function $g(t)\triangleq-t\ln(t)$ is monotone increasing in $0< t< 1/e$ and its inverse restricted to this interval satisfies
\begin{align}
\frac{1}{e}\cdot\frac{t}{-\ln(t)} < g^{-1}(t)\leq \frac{t}{-\ln(t)}.
\end{align}
It therefore follows that
\begin{align}
r\leq e g^{-1}\left(\frac{2I(X;Y)}{e\log(e)}\right)\leq \frac{2I(X;Y)}{\log(e)}\frac{1}{\ln\left(\frac{e\log(e)}{2I(X;Y)}\right)}
\end{align}
which gives
\begin{align}
I(X;f(Y))\leq 2M\frac{I(X;Y)}{\ln\left(\frac{e\log(e)}{2I(X;Y)}\right)},
\end{align}
for any $M$-level function $f$.
\end{proof}

\subsection{Proof of Theorem~\ref{thm:bindistillation}}\label{sec:proof}

We begin by proving the lower bound. Using Theorem~\ref{thm:binarylb} and solving for $\eta$, we obtain that
\begin{align}
I(X;[Y]_M)\geq \left(1-\frac{52\max\left\{\log\left(\frac{1}{\beta}\right),1\right\}}{M}\right)\cdot\beta.\label{eq:binfrac1}
\end{align}
As a consequence of Theorem~\ref{thm:binarylb}, we also have that
\begin{align}
I\left(X;[Y]_{\left\lfloor 104\max\left\{\log\left(\frac{1}{\beta}\right),1\right\}\right\rfloor}\right)\geq \frac{1}{2}\beta.
\end{align}
Now, applying Corollary~\ref{cor:doublequant}, we obtain that for any $M<104\max\left\{\log\left(\frac{1}{\beta}\right),1\right\}$ it holds that
\begin{align}
I(X;[Y]_M)\geq  \frac{M-1}{104\max\left\{\log\left(\frac{1}{\beta}\right),1\right\}}\frac{1}{2}\beta.\label{eq:binfrac2}
\end{align}
Combining~\eqref{eq:binfrac1} and~\eqref{eq:binfrac2} establishes that $\mathrm{ID}_M(2,\beta)\geq f_{\text{lower}}\left(\frac{M-1}{\max\left\{\log\left(\frac{1}{\beta}\right),1\right\}}\right)$.

To establish the upper bound, we use Theorem~\ref{thm:binaryub}. Note first that for any $M>1$ and $\beta>1/2$ we have that $f_{\text{upper}}\left(\frac{M-1}{\max\left\{\log\left(\frac{1}{\beta}\right),1\right\}}\right)=1$. Thus, it suffices to prove that $\mathrm{ID}_M(2,\beta)\leq f_{\text{upper}}\left(\frac{M-1}{\max\left\{\log\left(\frac{1}{\beta}\right),1\right\}}\right)$ for $\beta<1/2$. In this case Theorem~\ref{thm:binaryub} shows that
\begin{align}
\mathrm{ID}_M(2,\beta)&\leq 2M\frac{\beta}{\ln\left(\frac{e\log(e)}{2\beta}\right)}\nonumber\\
&=\frac{2}{\ln 2}\frac{M\beta}{\log\left(\frac{1}{\beta}\right)\left(1+\frac{\log(e\log(e)/2)}{\log\left(\frac{1}{\beta}\right)}\right)}\nonumber\\
&\leq \frac{2}{\ln 2\left(1+\log(e\log(e)/2)\right)}\cdot \frac{M}{\log\left(\frac{1}{\beta}\right)}\cdot\beta\nonumber\\
&< \frac{3}{2}\cdot \frac{M}{\log\left(\frac{1}{\beta}\right)}\cdot\beta.\label{eq:binubsimplified}
\end{align}
Combining this with the trivial bound $\mathrm{ID}_M(2,\beta)\leq \beta$, we obtain
\begin{align}
\mathrm{ID}_M(2,\beta)&\leq\min\left\{\frac{3}{2}\cdot \frac{M}
{\log\left(\frac{1}{\beta}\right)},1\right\}\cdot\beta\nonumber\\
&\leq \min\left\{\frac{3}{2}\cdot \frac{M}{\max\left\{\log\left(\frac{1}{\beta}\right),1\right\}},1\right\}\cdot\beta.
\end{align}
Noting that for $M=1$ we trivially have $\mathrm{ID}_M(2,\beta)=0$ and that $3M/2\leq 3(M-1)$ for $M>1$, we obtain the desired result.

\section{Bounds for $|\m{X}| > 2$}
\label{sec:nonbinary}

In the previous section we have shown that if $X$ is binary, 
then $\m{O}(\log(1/I(X;Y)))$ quantization levels always suffice in order to retain any constant fraction $0<\eta<1$ of $I(X;Y)$. In this section, we leverage this result in order to show that in general $\m{O}(\log(1/I(X;Y))^{k})$ quantization levels suffice in order to retain a constant fraction $0<\eta<\frac{k}{|\m{X}|-1}$, for $k\in[|\m{X}|-1]$.

For a random variable $X\in\{1,\ldots,|\m{X}|\}$, we can define the $|\m{X}|-1$ binary random variables $A_i \triangleq \mathbb{1}_{\{X = i\}}$, $i=1,\ldots,|\m{X}|-1$. Clearly, $X$ fully determines $\{A_1,\ldots,A_{|\m{X}|-1}\}$ and vice versa. In particular, the encoding of $X$ by $\{A_1,\cdots,A_{|\m{X}|-1}\}$ can be thought of as the ``one-hot'' encoding of $X$, with the last bit, whose value is deterministically dictated by the preceding $|\m{X}|-1$ bits, omitted. Representing $X$ in this manner, nevertheless, allows us to reduce the problem of quantizing $Y$ in order to retain information on $X$, into $|\m{X}|-1$ separate problems of quantizing $Y$ in order to retain information on $A_i$. Since the random variables $\{A_1,\ldots,A_{|\m{X}|-1}\}$ are binary, the results from Theorem~\ref{thm:binarylb} can be applied.

The main result of this section is Theorem~\ref{thm:nonbinary}, that lower bounds the worst-case multiplicative loss due to quantization. Before stating this result, and giving its proof,  we demonstrate the technique of reducing to the 
binary case via ``one-hot'' encoding for the setup considered in~\cite{kt17,tal15}. 
Recall that for all $M\geq 2|\m{X}|$ and $|\m{Y}|>2|\m{X}|$,~\cite[Theorem 1]{kt17} bounds the worst-case additive gap due to quantization as
\begin{align}
\sup_{P_{XY}} I(X;Y)-I(X;[Y]_M) \le \nu(|\m{X}|) \cdot M^{-\frac{2}{|\m{X}| - 1}}
\end{align}
where the function $\nu(|\m{X}|)$ is defined as in~\eqref{NuDefn} and the supremum is with respect to all $P_{XY}$ with input alphabet of cardinality $|\m{X}|$, and output alphabet of cardinality $|\m{Y}|$. We further note that $\nu(|\m{X}|)$ satisfies $\nu(2)\leq 1268$ and $\nu(|\m{X}|)\approx 16\pi e |\m{X}|^3$ for large $|\m{X}|$. 

Below, we use a ``one-hot'' encoding technique combined with~\eqref{eq:ktbound} for $| \m{X}| = 2$ only to obtain a slight refinement of the constant in the additive gap for $| \m{X} | > 2$ (for large enough values of $M$).

\begin{proposition}
For $|\m{X}|\geq 2$ and any $M$ such that $M^{\frac{1}{|\m{X}|-1}}\geq 4$ is an integer, we have
\begin{align}
\mathrm{DC}(|\m{X}|,M)\leq 1268(|\m{X}|-1)\cdot M^{-\frac{2}{|\m{X}|-1}} \ . \label{eq:DCimproved}
\end{align}
\end{proposition} 

\begin{proof}
Without loss of generality we may assume $|\m{Y}|> 4$, as otherwise the assumption $M^{\frac{1}{|\m{X}|-1}}\geq 4$ implies that $M\geq |\m{Y}|$, in which case $I(X;[Y]_M)=I(X;Y)$.

The case $|\m{X}|=2$ is therefore obtained from~\eqref{eq:ktbound}, which reads
\begin{align}
I(X;Y)-I(X;[Y]_M) \le \nu(2) \cdot M^{-2} \label{eq:additivebinary}
\end{align}
for all $P_{XY}$ with binary $X$. 
	
Now let $|\m{X}| > 2$, 
and without loss of generality assume $\m{X}=\{1,2,\ldots,|\m{X}|\}$. Define $A_i \triangleq \mathbb{1}_{\{X = i\}}$, for $i=1,2,\ldots,|\m{X}|-1$. Then, 
	\begin{align}
	I(X;Y) &= I(A_1,\ldots,A_{|\m{X}|-1};Y) \nonumber \\
	&= \sum_{i=1}^{|\m{X}|-1} I(A_i;Y|A_1^{i-1}=0)\Pr(A_1^{i-1}=0) \label{eq:additiveone}
	\end{align}
	where $A_1^{i-1}=0$ denotes the event $A_1  = \cdots = A_{i-1} = 0$. 
	
	Let $f(y)$ be an $M$-level quantizer of the form $f(y) = (f_1(y),\ldots,f_{|\m{X}|-1}(y))$. Then,
	\begin{align}
	&I(X;f(Y)) =\sum_{i=1}^{|\m{X}|-1} I(A_i;f(Y)|A_1^{i-1}=0)\Pr(A_1^{i-1}=0)\nonumber\\
	&\ge \sum_{i=1}^{|\m{X}|-1} I(A_i;f_i(Y)|A_1^{i-1}=0)\Pr(A_1^{i-1}=0). \label{eq:additivetwo}
	\end{align}
	Thus, combining~\eqref{eq:additiveone} and~\eqref{eq:additivetwo}, gives
	\begin{align}
	I(X;Y)-I(X;f(Y))
	\le \sum_{i=1}^{|\m{X}|-1}& \big(I(A_i;Y|A_1^{i-1}=0)  \nonumber \\ & - I(A_i;f_i(Y)|A_1^{i-1}=0)\big)\Pr(A_1^{i-1}=0).
	\end{align}
	From~\eqref{eq:additivebinary}, it holds that by choosing $|f_i(y)| = M^{\frac{1}{|\m{X}|-1}}\geq 4$, for all $1 \le i \le |\m{X}|-1$, we can find quantizers
	 $f_1(y),\ldots,f_{|\m{X}|-1}(y)$ for which
	\begin{align}
	I(A_i;Y|A_1^{i-1}=0)-I(A_i;f_i(Y)|A_1^{i-1}=0) \le \nu(2)\cdot M^{\frac{-2}{|\m{X}|-1}}. \label{eq:chainrule}
	\end{align}
	Consequently, with this choice, we obtain
	\begin{align}
	I(X;Y)-I(X;f(Y)) &\le \nu(2)\cdot M^{\frac{-2}{|\m{X}|-1}} \sum_{i=1}^{|\m{X}|-1}\Pr(A_1^{i-1}=0)\nonumber\\
	&\le (|\m{X}|-1)\nu(2) \cdot M^{\frac{-2}{|\m{X}|-1}},
	\end{align}
	as desired.
	\end{proof}

Next, we focus on the regime of $I(X;Y)\ll 1$ and prove an upper bound on the number of quantization levels $M$, required to attain a fraction $0<\eta<1$ of $I(X;Y)$. Roughly, we show that it suffices to take $M$ that scales like $(\log(1/I(X;Y)))^{\eta\cdot(|\m{X}|-1)}$. More precisely, for any $k\in[|\m{X}|-1]$, if $\eta<\frac{k}{|\m{X}|-1}$, then $\m{O}\left((\log(1/I(X;Y)))^{k}\right)$ levels suffice.
	
\begin{theorem}
For any $P_{XY}$ with $I(X;Y)=\beta $, and any $\eta\in(0,1)$, we have that
\begin{align}
I(X;[Y]_{\bar{M}_{|\m{X}|}(\eta,\beta)})\geq \eta\beta,
\end{align}
where
\begin{align}
\bar{M}_{|\m{X}|}(\eta,\beta) &\triangleq 
\left\lfloor\left[c_1\left(\sqrt{\eta}\right)\log\left(\frac{|\m{X}|-1}{(1-\sqrt{\eta})\beta} \right)\right]^{|\m{X}|-1}\right\rfloor,\label{eq:Mbardef}
\end{align}
with $c_1(\eta)$ as defined in~\eqref{eq:c1}. Furthermore, if $\eta<\frac{k}{|\m{X}|-1}$ for some natural $k\leq |\m{X}|-2$, we have that
\begin{align}
I(X;[Y]_{\tilde{M}_{|\m{X}|}(\eta,\beta,k)})\geq \eta\beta,
\end{align}
where
\begin{align}
\tilde{M}_{|\m{X}|}(\eta,\beta,k) &\triangleq 
\left\lfloor\left[c_1\left(\frac{\eta}{k/(|\m{X}|-1)}\right)\log\left(\frac{(|\m{X}|-1)^2}{\beta} \right)\right]^k\right\rfloor.
\end{align}
\label{thm:nonbinary}
\end{theorem}

\begin{proof}
As above, define $A_i \triangleq \mathbb{1}_{\{X = i\}}$, such that 
\begin{align}
I(X;Y) &=I(A_1,\ldots,A_{|\m{X}|-1};Y) =\sum_{i=1}^{|\m{X}|-1} I_i\cdot p_i \label{eq:Ixysum}
\end{align}
where
\begin{align}
I_i&\triangleq I(A_i;Y|A_1^{i-1}=0), \ i=1,\ldots,|\m{X}|-1, \\
p_i&\triangleq \Pr(A_1^{i-1}=0), \ \ \ \ \ \ \ i=1,\ldots,|\m{X}|-1,
\end{align}
and $A_1^{i-1}=0$ denotes the event $A_1  = \cdots = A_{i-1} = 0$. Furthermore, set 
\begin{align}
v_i\triangleq \frac{I_i\cdot p_i}{\beta},  \ i=1,\ldots,|\m{X}|-1,
\end{align}
and let the permutation $\pi:[|\m{X}|-1]\to [|\m{X}|-1]$ be such that $v_{\pi(1)}\geq\cdots\geq v_{\pi(|\m{X}|-1)}$.
For $0\leq k\leq |\m{X}|-1$, define the function 
\begin{align}
F(k)\triangleq\sum_{i=1}^k v_{\pi(i)},
\end{align}
with the convention that $F(0)=0$, and note that 
\begin{enumerate}
\item $F(t)\geq \frac{t}{|\m{X}|-1}$ for any natural $t\leq |\m{X}|-1$, and in particular, $F(|\m{X}|-1)=1$;
\item $F(|\m{X}|-1)-F(t-1)=\sum_{i=t}^{|\m{X}|-1}v_{\pi(i)}\leq (|\m{X}|-t)v_{\pi(t)}$ and therefore, for any natural $t\leq|\m{X}|-1$ we have that $v_{\pi(t)}\geq\frac{1-F(t-1)}{|\m{X}|-t}$;
\end{enumerate}

Let $\eta\in(0,1)$ and let $\bar{\eta}$ be some number satisfying
$0<\eta<\bar{\eta}<1$. Let
\begin{align}
k_{\bar{\eta}}=\min\{k \ : \ F(k)\geq\bar{\eta} \}.
\end{align}
By the definition of $k_{\bar{\eta}}$, we have that $F(k_{\bar{\eta}}-1)< \bar{\eta}$. Thus, by the second property, we have that
\begin{align}
v_{\pi(k_{\bar{\eta}})}\geq \frac{1-\bar{\eta}}{|\m{X}|-k_{\bar{\eta}}}\geq \frac{1-\bar{\eta}}{|\m{X}|-1}.\label{eq:vklb}
\end{align}
Let $\eta'=\frac{\eta}{\bar{\eta}}<1$. 
Consider the conditional joint distribution $P_{A_i Y|A_1^{i-1}=0}$. Since $A_i$ is a binary random variable, by Theorem~\ref{thm:binarylb} we can design a quantizer $f_{i}:\m{Y}\to [M_i]$ with $M_i\leq \big\lfloor c_1(\eta')\log\max\big\{\big(\frac{1}{I_i} \big),1\big\}\big\rfloor$ quantization levels, such that $I(A_i;f_i(Y)|A_{1}^{i-1}=0)\geq \eta'\cdot I_i$. Let $f(y)=(f_{\pi(1)}(y),\ldots,f_{\pi(k_{\bar{\eta}})}(y)):\m{Y}\to[M_{\pi(1)}]\times\cdots\times[M_{\pi(k_{\bar{\eta}})}]$ be the Cartesian product of the quantizers $f_{\pi(1)}(y),\ldots,f_{\pi(k_{\bar{\eta}})}(y)$ attaining this tradeoff between $\eta'$ and the number of quantization levels. We have that
\begin{align}
I(X&;f(Y))=I(A_1,\ldots,A_{|\m{X}|-1};f(Y))\nonumber\\
&=\sum_{i=1}^{|\m{X}|-1}I(A_i;f(Y)|A_1^{i-1}=0)\Pr(A_1^{i-1}=0)\nonumber\\
&\geq\sum_{i=1}^{k_{\bar{\eta}}}I(A_{\pi(i)};f(Y)|A_1^{\pi(i)-1}=0)\Pr(A_1^{\pi(i)-1}=0)\nonumber\\
&\geq\sum_{i=1}^{k_{\bar{\eta}}}I(A_{\pi(i)};f_{\pi(i)}(Y)|A_1^{\pi(i)-1}=0)p_{\pi(i)}\nonumber\\
&\geq \sum_{i=1}^{k_{\bar{\eta}}}\eta' I_{\pi(i)}p_{\pi(i)}\nonumber\\
&=\eta'\beta \sum_{i=1}^{k_{\bar{\eta}}} v_{\pi(i)}\nonumber\\
&=  \frac{\eta}{\bar{\eta}}\beta F(k_{\bar{\eta}})\nonumber\\
&\geq \eta\beta.
\end{align}
Since $I_{\pi(i)}=\beta v_{\pi(i)}/p_{\pi(i)}\geq \beta v_{\pi(i)}>\frac{\beta(1-\bar{\eta})}{|\m{X}|-1}$, $\forall i\leq k_{\bar{\eta}}$, by~\eqref{eq:vklb}, we have that $\forall i\leq k_{\bar{\eta}}$
\begin{align}
M_{\pi(i)}&\leq \left\lfloor c_1(\eta')\max\left\{\log\left(\frac{|\m{X}|-1}{(1-\bar{\eta})\beta} \right),1\right\}\right\rfloor\nonumber\\
&=\left\lfloor c_1(\eta')\log\left(\frac{|\m{X}|-1}{(1-\bar{\eta})\beta} \right)\right\rfloor,
\end{align}
where the last inequality follows since $\frac{|\m{X}|-1}{\beta}\geq
\frac{|\m{X}|-1}{\log|\m{X}|}\geq 2$ for all $|\m{X}|> 2$.
Consequently, we obtained
\begin{align}
|f(y)|\leq \left\lfloor\left[c_1\left(\frac{\eta}{\bar{\eta}}\right)\log\left(\frac{|\m{X}|-1}{(1-\bar{\eta})\beta} \right)\right]^{k_{\bar{\eta}}}\right\rfloor.\label{eq:fycard}
\end{align}
To establish the  first part of the statement, take $\bar{\eta}=\sqrt{\eta}$ and recall that $k_{\bar{\eta}}\leq |\m{X}|-1$ by definition.

For the second part, note that if $\eta<\frac{k}{|\m{X}|-1}$, for some $k<|\m{X}|-1$ we may take $\bar{\eta}=\frac{k}{|\m{X}|-1}$, and that $k_{\bar{\eta}}\leq k$, as $F(t)\geq\frac{t}{|\m{X}|-1}$. Substituting into~\eqref{eq:fycard}, and noting that $1-\bar{\eta}\geq \frac{1}{|\m{X}|-1}$, establishes the second part of the statement.
\end{proof}

Theorem~\ref{thm:nonbindistillation} now follows as a rather simple corollary.

\begin{proof}[Proof of Theorem~\ref{thm:nonbindistillation}]
We first show that for $1\leq k\leq |\m{X}|-2$, it holds that $I(X;[Y]_M)\geq a_k(M,|\m{X}|,\beta)\cdot\beta$.
To that end, for any $0<\eta'<1$, and $1\leq k\leq |\m{X}|-2$ define
\begin{align}
M(|\m{X}|,\beta,\eta',k)=\left\lfloor\left[\frac{52}{1-\eta'}\log\left(\frac{(|\m{X}|-1)^2}{\beta} \right)\right]^k\right\rfloor.
\end{align}
By the second part of Theorem~\ref{thm:nonbinary}, we have that
\begin{align}
I(X;[Y]_{M(|\m{X}|,\beta,\eta',k)})\geq \frac{k}{|\m{X}|-1}\eta'\beta.\label{eq:kfrac}
\end{align}
Solving for $\eta'$ shows that 
\begin{align}
I(X;[Y]_M)\geq \frac{k}{|\m{X}|-1} \left(1-\frac{52\log\left(\frac{(|\m{X}|-1)^2}{\beta} \right)}{M^{\frac{1}{k}}}\right)\beta,
\end{align}
and maximizing with respect to $k$ yields the bound
\begin{align}
I(X;[Y]_M)\geq\max_{1\leq k\leq |\m{X}|-2}a_k (M,|\m{X}|,\beta)\cdot\beta.\label{eq:akbound}
\end{align}
Moreover,~\eqref{eq:kfrac} applied with $\eta'=1/2$ and $k=1$ shows that
\begin{align}
I\left(X;[Y]_{\left\lfloor 104 \log\left(\frac{(|\m{X}|-1)^2}{\beta} \right)\right\rfloor}\right)\geq \frac{1}{2(|\m{X}|-1)}\beta.
\end{align}
Now, applying Corollary~\ref{cor:doublequant}, we obtain that for any $M<104 \log\left(\frac{(|\m{X}|-1)^2}{\beta} \right)$ it holds that
\begin{align}
I(X;[Y]_M)\geq\frac{M-1}{104 \log\left(\frac{(|\m{X}|-1)^2}{\beta} \right)}\frac{1}{2(|\m{X}|-1)}\beta.
\end{align}
which is equivalent to
\begin{align}
I(X;[Y]_M)\geq a_0(M,|\m{X}|,\beta)\cdot\beta.\label{eq:a0bound}
\end{align}
Finally, we use the first part of Theorem~\ref{thm:nonbinary} to show that $I(X;[Y]_M)\geq a_{|\m{X}|-1} (M,|\m{X}|,\beta)\cdot\beta$. Recalling the definition of $\bar{M}_{|\m{X}|}(\eta,\beta)$ in~\eqref{eq:Mbardef}, we have that for any $0<\eta<1$
\begin{align}
&\bar{M}_{|\m{X}|}(\eta,\beta)=\left\lfloor\left[c_1\left(\sqrt{\eta}\right)\log\left(\frac{|\m{X}|-1}{(1-\sqrt{\eta})\beta} \right)\right]^{|\m{X}|-1}\right\rfloor\nonumber\\
&\leq \left\lfloor\left[c_1(\sqrt{\eta})\left(\log\left(\frac{|\m{X}|-1}{\beta}\right)+\frac{\log(e)}{(1-\sqrt{\eta})^{1/2}}\right)\right]^{|\m{X}|-1}\right\rfloor	\nonumber\\
&\leq \left\lfloor\left[\frac{52}{(1-\sqrt{\eta})^{3/2}}\left(\log\left(\frac{e(|\m{X}|-1)}{\beta}\right)\right)\right]^{|\m{X}|-1}\right\rfloor\nonumber\\
&\triangleq {M}'_{|\m{X}|}(\eta,\beta)
\end{align}
where the first inequality follows since $\log\left(\frac{1}{x}\right) \le \frac{\log(e)}{\sqrt{x}}$ for $0 \le x \le 1$. This can be established by differentiating the function $\frac{1}{\sqrt{x}} - \ln\left(\frac{1}{x}\right)$ and finding that $\min_{0 < x \le 1} \frac{1}{\sqrt{x}} - \ln\left(\frac{1}{x}\right) > 0$. Thus, by the first part of Theorem~\ref{thm:nonbinary}, we have that
\begin{align}
I(X;[Y]_{{M}'_{|\m{X}|}(\eta,\beta)})\geq \eta\beta.
\end{align}
Solving for $\eta$ yields
\begin{align}
I(X;[Y]_M)\geq a_{|\m{X}|-1} (M,|\m{X}|,\beta)\cdot\beta.\label{eq:axbound}
\end{align}
The theorem now follows by combining~\eqref{eq:akbound},~\eqref{eq:a0bound}, and~\eqref{eq:axbound}.
\end{proof}

\section{Connections to Quantization Under Log-Loss and the Information Bottleneck Problem}

In general, an $M$-level quantizer $q$ for a random variable $Y$ consists of a disjoint partition of its alphabet $\m{Y}=\bigcup_{i=1}^M\m{S}_i$, and a set of corresponding reproduction values $a_i\in\m{A}$, such that $q_y=\sum_{i=1}^M a_i\Ind_{\{y\in\m{S}_i\}}$, see, e.g.,~\cite{gn98}. The performance of the quantizer is measured with respect to some predefined distortion function $d:\m{Y}\times\m{A}\to\RR$, which quantifies the ``important features'' of $Y$ that the quantizer should aim to retain. The expected distortion $\mathbb{E}d(Y,q_Y)$ is then typically taken as the quantizer's main figure of merit.

In our considerations, we observe and quantize the random variable $Y$, but the distortion measure is evaluated with respect to $X$, where $X$ and $Y$ are jointly distributed according to $P_{XY}$. This setup is sometimes referred to as  remote source coding (or quantization). If the distortion measure of interest between $X$ and the reconstruction $q_y$ is $\tilde{d}:\m{X}\times\m{A}\to\RR$, one can define the induced distortion measure
\begin{align}
d(y,q_y)=\mathbb{E}\left[\tilde{d}(X,q_y)|Y=y\right],
\end{align}
such that $\mathbb{E}[\tilde{d}(X,q_Y)]=\mathbb{E}[d(Y,q_Y)]$. Consequently, the remote quantization problem is reduced to a direct quantization problem, with an induced distortion measure~\cite{dt62}.

Under various tasks of inferring information about $X$ from $Y$, it is natural to take the reconstruction alphabet $\m{A}$ to be the set of all distributions on $\m{X}$, i.e., the $|\m{X}|-1$ dimensional simplex $\m{P}^{|\m{X}|-1}$~\cite{as66,csiszar67,mf98} . Ideally, we would like the reconstructed distribution $q_y$ to be as close as possible to the conditional distribution $P_{X|Y=y}$, for all $y\in\m{Y}$. Various loss functions can be used to measure the distance between two distributions, depending on the ultimate performance criterion for the inference of $X$. One such loss function, that has enjoyed a special status in the information theory and machine learning literature~\cite{mf98,cw14,jcvw15,cesabianchilugosi,sv18,srv18} is the logarithmic-loss:
\begin{align}
d(x,P)=\log\left(\frac{1}{P(x)}\right), \ 
\forall (x,P)\in\m{X}\times\m{P}^{|\m{X}|-1}.
\end{align}
For the remote quantization setup, the induced distortion measure is
\begin{align}
d(y,P)\triangleq \mathbb{E}\left[ \log\frac{1}{P(X)}\bigg| \, Y=y\right], \ \forall (y,P)\in\m{Y}\times\m{P}^{|\m{X}|-1}.\label{eq:distdef}
\end{align}
Thus, the design of a quantizer for $Y$ under $d(y,P)$ reduces to determining a disjoint partition $\m{Y}=\bigcup_{i=1}^M\m{S}_i$ of the alphabet $\m{Y}$, and assigning a representative distribution $a_i\in\m{P}^{|\m{X}|-1}$ for each quantization cell $S_i$, such that $q_y=a_i$ iff $y\in S_i$. Note that once the sets $S_i$, $i=1,\ldots,M$ are determined, the reconstructions that minimize $D=\mathbb{E}d(Y,q_Y)$ are given by $a_i=P_{X|Y\in\m{S}_i}$. To see this, let $f:\m{Y}\to[M]$ be such that $f(y)=i$ if $i\in\m{S}_i$, set
$T=f(Y)$, and write
\begin{align}
D&=\mathbb{E}_{XY}\left[\log\left(\frac{1}{q_Y(X)}\right)\right]\nonumber\\
&=\mathbb{E}_{XT}\left[\log\left(\frac{1}{a_T(X)}\right)\right]\nonumber\\
&=\mathbb{E}_T\left[\mathbb{E}\left[\log\left(\frac{1}{P_{X|T}(X|T)}\frac{P_{X|T}(X|T)}{a_T(X)}\right) \bigg|T\right]\right]\nonumber\\
&=H(X|T)+D\big(P_{X|T}\big\| a_T \big| P_T\big)\nonumber\\
&\geq H(X|T)\nonumber\\
&=H(X|f(Y)),
\end{align}
with equality if and only if $a_t=P_{X|T=t}=P_{X|Y\in\m{S}_t}$ for all $t\in[M]$.

It follows that, for a given distribution $P_{XY}$, the design of the optimal quantizer under the distortion measure~\eqref{eq:distdef} reduces to finding $f: \m{Y} \rightarrow [M]$ which minimizes $H(X|f(Y))$. Clearly, determining the minimum value of $H(X|f(Y))$ is equivalent to our maximization problem~\eqref{eq:gmdef}.

A quantity closely related to $I(X;[Y]_M)$ is the information bottleneck tradeoff~\cite{tpb99}, defined as
\begin{align}
\mathrm{IB}_R(P_{XY})\triangleq\max_{P_{T|Y} \, : \, I(Y;T)\leq R} I(X;T),\label{eq:IB}
\end{align}
which has been extensively studied in the machine learning literature, see e.g.~\cite{st2000,tz15,st17}.
There, $Y$ is thought of as a high-dimensional observation containing information about $X$, that must be first ``compressed'' to a simpler representation before inference can be efficiently performed. The random variable $T=f(Y)$ represents a clustering operation, where for the task of inferring $X$, all members in the cluster are treated as indistinguishable. A major difference, however, between the information bottleneck formulation and that of~\eqref{eq:gmdef} is that the latter restricts $|f(\cdot)|$ to $M$, whereas the former allows for random quantizers and restricts the compression rate $I(T;Y)$. The discussion above indicates that the problem~\eqref{eq:gmdef} is a standard quantization/lossy compression problem (or more precisely, a remote source coding problem). As such, its fundamental limit admits a single-letter solution\footnote{One subtle point to be noted is that the relevant distortion measure for $I(X^n;[Y]^n_{M^n})$ is not separable. Nevertheless, it is not difficult to show that restricting the reconstruction distribution to the form $q_{y^n}(x^n)=\prod_{i=1}^n q^i_{y^n}(x_i)$ entails no loss asymptotically.} and we have that~\cite{gnt03,cw14}
\begin{align}
\lim_{n\to\infty} \frac{1}{n} I(X^n;[Y^n]_{M^n})=\mathrm{IB}_{\log{M}}(P_{XY}).
\end{align} where $P_{X^nY^n} = P^{\otimes n}_{XY}$ and $[Y^n]_{M^n}$ refers to the set of all $M^n$-quantizations of $Y^n$. That is, in the asymptotic limit, our problem~\eqref{eq:gmdef} corresponds to an information bottleneck problem. However, the scalar setting $n=1$ is of major importance as inference is seldom performed in blocks consisting of multiple independent samples from $P_{XY}$. Overall, our results in the previous sections indicate that when $I(X;Y)$ is small we may need at least $\Theta(\log(1/I(X;Y))$ clusters to guarantee that we retain a significant fraction of the original information.

\subsection{On the Gap Between Scalar Quantization and Information Bottleneck}

In this subsection, we show that in the limit $I(X;Y)\to 0$, the restriction to a scalar quantizer results in significantly worse performance than that predicted by the information bottleneck, which implicitly assumes quantization is performed in asymptotically large blocks. In particular, we prove the following theorem.
\begin{theorem}
For any $P_{XY}$ with $|\m{X}|=2$ and $I(X;Y)=\beta$, and any $\eta\in (0,1)$ there exists a quantizer $f(Y)$ such that $I(X;f(Y))\geq \eta\beta$ and
	\begin{align}
	H(f(Y))\leq \log\log\log\left(\frac{1}{\beta}\right)-\log(1-\eta)+\log(-\log(1-\eta))+11.
	\end{align}
	\label{thm:entcoding}
\end{theorem}

Contrasting this with Theorem~\ref{thm:main}, and its simplification in~\eqref{eq:binubsimplified}, which show that there exist distributions $P_{XY}$ with $|\m{X}|=2$, for which no scalar quantizer with less than $\log\log(1/\beta)+\log(\eta)+1$ bits can attain $I(X;f(Y))>\eta\beta$, we see that the restriction to quantization in blocklength $n=1$ entails a significant cost with respect to quantization in long blocks. In particular, if for a distribution $P_{XY}$ there exists a quantizer $f(Y)$ with entropy $H(f(Y))=R$ for which $I(X;f(Y))=\Gamma$, then certainly $\mathrm{IB}_R(P_{XY})\geq \Gamma$. To see this just take $T=f(Y)$ in~\eqref{eq:IB}.\footnote{See~\cite{ss16} for an elaborate discussion on the information bottleneck tradeoff when $T$ is restricted to be a deterministic quantizer of $Y$.} It therefore follows from Theorem~\ref{thm:main} and Theorem~\ref{thm:entcoding} that the information bottleneck tradeoff may be over-optimistic in predicting the performance of optimal scalar quantization.

\begin{proof}
	In the proof of the lower bound of Theorem~\ref{thm:main}, we have proposed the $M$-level quantizer~\eqref{eq:uniquant} with the parameters specified by~\eqref{eq:thetarule}. For $M=\lfloor\frac{52}{1-\eta}\log\big(\frac{1}{\beta}\big)\rfloor$, and $\beta\leq\frac{1-\eta}{2}$, we have shown that this quantizer attains $I(X;f(Y))\geq \eta\beta$. We will now show that for the same quantizer $H(f(Y))=\m{O}\left(\log\log(M)\right)$.
	
	Let
	\begin{align}
	P_{\ell}\triangleq \Pr\left(\{f(Y)=-\ell\}\cup\{f(Y)=\ell\}\right), \ \ell=0,\ldots,L\nonumber
	\end{align}
	and note that 
	\begin{align}
	H(f(Y))\leq 1+H(\{P_\ell\}).\label{eq:entf}
	\end{align}
	 Our goal is therefore to derive universal upper bounds on $H(\{P_\ell\})$ that hold for all joint distributions $P_{XY}$ with $|\mathcal{X}|=2$.
	
	First, recall from the proof of Theorem~\ref{thm:binarylb} that
		\begin{align}
	I(X;Y)=\mathbb{E}D_Y\geq\sum_{\ell=0}^L \gamma_{\ell}P_{\ell}=\gamma_1\sum_{\ell=1}^L \nonumber \theta^{\ell-1}P_{\ell},
	\end{align}
	where we have used~\eqref{eq:thetarule} in the last equality. We therefore have
	\begin{align}
	\sum_{\ell=0}^L \theta^{\ell}P_{\ell}&=P_0+\sum_{\ell=1}^L
\theta^{\ell}P_{\ell}\nonumber\\
&\leq 1+\frac{\theta I(X;Y)}{\gamma_1}\nonumber\\
&=1+\frac{\theta}{(1-\eta)/2} \label{eq:beforePIconstraint}\\
&\leq \frac{4\theta}{1-\eta} \label{eq:Plconstraint}
	\end{align}
	where in~\eqref{eq:beforePIconstraint} we have used $\gamma_1=\epsilon I(X;Y)$, due to~\eqref{eq:thetarule}, and $\epsilon=(1-\eta)/2$, due to~\eqref{eq:epsdef}.
	
	For a vector $\ba=\{a_0,a_1,\ldots,a_{L}\}\in \RR_{+}^{L+1}$ and a scalar $\min_\ell\{a_\ell\}\leq b\leq\max_{\ell}\{a_{\ell}\}$, define the function
	\begin{align}
	f(\ba,b)\triangleq &\max \sum_{\ell=0}^L P_{\ell}\log \left(\frac{1}{P_{\ell}}\right) \nonumber\\
	&\text{subject to } \sum_{\ell=0}^L a_{\ell}P_{\ell}\leq b, \ \sum_{\ell=0}^L P_{\ell}=1.
	\label{eq:entmax}
	\end{align}
	The problem~\eqref{eq:entmax} is a concave maximization problem under linear constraints, and its solution is~\cite[p.228]{boydvandenberghe}
	\begin{align}
	f(\ba,b)=\min_{\lambda\geq 0}\lambda b+\log\left(\sum_{\ell=0}^L 2^{-\lambda a_\ell} \right).\label{eq:entmaxsol}
	\end{align}
	Combining~\eqref{eq:Plconstraint} and~\eqref{eq:entmaxsol} with $a_\ell=\theta^\ell$ and $b=\frac{4\theta}{1-\eta}$, gives
	\begin{align}
	H(\{P_{\ell}\})\leq \min_{\lambda\geq 0}\lambda\frac{4\theta}{1-\eta}+\log\left(\sum_{\ell=0}^L 2^{-\lambda \theta^\ell}\right).\label{eq:HplBound}
	\end{align}
	Setting $\lambda=\frac{1}{L}$, gives
	\begin{align}
	&H(\{P_{\ell}\})\leq \frac{4}{1-\eta}\cdot\frac{\theta}{L}+\log\left(\sum_{\ell=0}^L 2^{-\frac{\theta^{\ell}}{L}}\right)\nonumber\\
	&\leq\frac{4}{1-\eta}\cdot\frac{\theta}{L}+\log\left(\sum_{\ell=0}^{\lfloor 2\frac{\log{L}}{\log{\theta}} \rfloor} 2^{-\frac{\theta^{\ell}}{L}}+\sum_{\ell=\lfloor 2\frac{\log{L}}{\log{\theta}}\rfloor +1}^{L}2^{-\frac{\theta^{\ell}}{L}}\right)\nonumber\\
	&\leq \frac{4}{1-\eta}\cdot\frac{\theta}{L}+\log\left({2\frac{\log{L}}{\log{\theta}}}+1+L2^{-L}\right),\nonumber
	\end{align}
	where in the last transition we have used the fact that $2^{-\frac{\theta^{\ell}}{L}}\leq 1$ for all $0\leq\ell\leq \lfloor 2\frac{\log{L}}{\log{\theta}}\rfloor$ such that  $\sum_{\ell=0}^{\lfloor 2\frac{\log{L}}{\log{\theta}} \rfloor} 2^{-\frac{\theta^{\ell}}{L}}\leq {2\frac{\log{L}}{\log{\theta}}}+1$, and that $2^{-\frac{\theta^{\ell}}{L}}\leq 2^{-L}$ for all $\lfloor 2\frac{\log{L}}{\log{\theta}}\rfloor+1\leq \ell\leq L$ such that $\sum_{\ell=\lfloor 2\frac{\log{L}}{\log{\theta}}\rfloor +1}^{L}2^{-\frac{\theta^{\ell}}{L}}<L2^{-L}$.
	Recalling that $1<\theta<2^{(1-\eta)/2}$ due to~\eqref{eq:thetauba}, and noting that $1 + L2^{-L}<2\log(L)$ for $L>1$, we obtain
	\begin{align}
	&H(\{P_{\ell}\})\leq \frac{1}{L}\frac{2^{(1-\eta)/2}}{(1-\eta)/4} +\log\left(4\frac{\log{L}}{\log{\theta}}\right).\nonumber
	\end{align}
	For the first term, we can use the definition of $L$ in~\eqref{eq:epsdef} to obtain
	\begin{align}
	\frac{1}{L}\frac{2^{(1-\eta)/2}}{(1-\eta)/4}&\leq\frac{(1-\eta)/2}{\log\left(\frac{2\kappa}{(1-\eta)I(X;Y)}\right)}\cdot\frac{2^{(1-\eta)/2}}{(1-\eta)/4}\nonumber\\
	&\leq 2\cdot 2^{(1-\eta)/2}\nonumber\\
	&\leq 3,
	\end{align}
	where we have used the fact that $\kappa\geq 1$, and consequently $\log\left(\frac{2\kappa}{(1-\eta)I(X;Y)}\right)\geq 1$.
	 For the second term, we have that 
	\begin{align}
	\log(\theta)&\overset{(i)}{\geq} \frac{1}{L}\log\left(\frac{\kappa}{\epsilon \beta}\right)\nonumber\\
	&\overset{(ii)}\geq\frac{(1-\eta)/4}{\log\left(\frac{\kappa}{\beta\cdot (1-\eta)/2}\right)}\log\left(\frac{\kappa}{\beta\cdot (1-\eta)/2}\right)\nonumber\\
	&=\frac{1-\eta}{4}.
	\end{align} where $(i)$ follows from~\eqref{eq:thetarule} and $(ii)$ uses the following upper bound on the choice of $L$ from~\eqref{eq:epsdef}:
	\begin{align}
	L&=\left\lceil\frac{2}{1-\eta} \log\left(\frac{\kappa}{\beta\cdot (1-\eta)/2}\right)\right\rceil\nonumber\\
	&\leq \frac{4}{1-\eta} \log\left(\frac{\kappa}{\beta\cdot (1-\eta)/2}\right).\nonumber
	\end{align}
	We have therefore obtained that
	\begin{align}
	H(\{P_{\ell}\})
	&\leq 7+\log\log{L}-\log(1-\eta).\nonumber
	\end{align}
	Now, recalling that $L<M<\frac{52\log(1/\beta)}{1-\eta}$ (see beginning of proof),
	we have that
	\begin{align}
	H(\{P_{\ell}\})
	&\leq 10 + \log(-\log(1-\eta))-\log(1-\eta)+\log\log\log\left(\frac{1}{\beta}\right).\label{eq:Hplub}
	\end{align}
	 Now, applying~\eqref{eq:entf} with~\eqref{eq:Hplub} establishes the result.
\end{proof}

\section*{Acknowledgment}
The authors are grateful to Emre Telatar for pointing out the relation to the Knapsack problem, to Meir Feder for suggesting to compare the performance of scalar quantizers with that of scalar quantization followed by entropy coding for $f(Y)$, and to Guy Bresler, Robert Gray, Pablo Piantanida, and Shlomo Shamai (Shitz) for valuable discussions.

 \bibliographystyle{IEEEtran}
 \bibliography{IAQbib}

\end{document}